\documentclass[a4paper,10pt]{article}

\usepackage[english]{babel}
\usepackage[T1]{fontenc}
\usepackage[dvips]{graphicx}
\usepackage{amsmath, amsfonts,amsthm, mathrsfs, amssymb}
\usepackage{verbatim}

\usepackage{enumerate}
\usepackage[latin1]{inputenc}
\usepackage{bbm}
\usepackage{geometry}

\def\bq{{\bf q}}
\def\bee{{\bf e}}
\def\bp{{\bf p}}

\newcommand{\wtr}{{\widetilde{r}}}
\newcommand{\wts}{{\widetilde{s}}}

\newcommand{\wtP}{\widetilde{P}}
\newcommand{\wtQ}{\widetilde{Q}}

\newcommand{\wtS}{\widetilde{S}}

\def\wcK{\widetilde{\cK}}

\def\norma#1{\left\|#1\right\|}
\def\sleq{\preceq}

\newcommand{\wtcH}{{\widetilde{\mathcal H}}}

\def\U{{\mathcal U}}

\def\vphi{{\tt f}}

\newcommand{\C}{{\mathbb C}}

\newcommand{\N}{{\mathbb N}}

\newcommand{\R}{{\mathbb R}}

\newcommand{\T}{{\mathbb T}}

\newcommand{\cA}{{\mathcal A}}

\newcommand{\cD}{{\mathcal D}}
\newcommand{\cE}{{\mathcal E}}
\newcommand{\cF}{{\mathcal F}}

\newcommand{\cH}{{\mathcal H}}

\newcommand{\cJ}{{\mathcal J}}
\newcommand{\cK}{{\mathcal K}}

\newcommand{\cO}{{\mathcal O}}
\newcommand{\cP}{{\mathcal P}}

\newcommand{\cR}{{\mathcal R}}
\newcommand{\cS}{{\mathcal S}}
\newcommand{\cT}{{\mathcal T}}
\newcommand{\cU}{{\mathcal U}}
\newcommand{\cV}{{\mathcal V}}
\newcommand{\cW}{{\mathcal W}}

\newcommand{\fh}{{\mathfrak{h}}}

\newcommand{\fr}{{\mathfrak{r}}}

\renewcommand{\d}{\partial}
\newcommand{\dt}{\partial_t}
\newcommand{\lap}{\Delta}
\newcommand{\grad}{\nabla}

\newcommand{\norm}[1]{\| #1 \|}
\newcommand{\normk}[2]{\| #1 \|_{\mathcal{H}^{#2}}}
\newcommand{\mmod}[1]{\left| #1 \right|}

\newcommand{\la}{\left\langle}
\newcommand{\ra}{\right\rangle}

\newcommand{\id}{\mathbbm{1}}
\newcommand{\im}{{\rm i}}
\newcommand{\jap}[1]{\langle #1 \rangle}

\newcommand{\be}{b_{\mathcal{E}}}

\newcommand{\Pip}{\varPi_p}

\newcommand{\Pipj}{\frac{\partial \varPi_p}{\partial p_j}}
\newcommand{\Pipk}{\frac{\partial \varPi_p}{\partial p_k}}

\newcommand{\Pipz}{\varPi_{0}}

\newcommand{\etap}{\eta_p}

\newcommand{\etapj}{\frac{\partial \eta_p}{\partial p_j}}
\newcommand{\etapk}{\frac{\partial \eta_p}{\partial p_k}}

\newcommand{\SH}{\mathcal{S}}
\newcommand{\SR}{\mathcal{R}}

\newcommand{\ALS}{\mathcal{A}\ell \mathcal{S}}

\newcommand{\Tr}{\mathscr{T}}

\def\ap{ {\mathfrak{A}} }

\def\V{{\mathcal V}}
\def\e{{\rm e}}
\def\di{{\rm d}}
\def\uno{\mathbbm{1}}
\def\dep#1{\frac{\partial\Pip}{\partial p_{#1}}}
\def\depv#1#2{\frac{\partial #1}{\partial p_{#2}}}
\def\dip#1#2{\frac{\partial#1}{\partial p_{#2}}}
\numberwithin{equation}{section}
\newtheorem{theorem}{Theorem}[section]
\newtheorem{lemma}[theorem]{Lemma}
\newtheorem{corollary}[theorem]{Corollary}
\newtheorem{proposition}[theorem]{Proposition}
\newtheorem{definition}[theorem]{Definition}
\newtheorem{remark}[theorem]{Remark}

\def\intre{\int_{\R^3}}
\def\resto{{\mathcal R}}
\def\Sc{{\mathcal S}}
\def\coS{{\mathcal F}}
\def\ex#1{e^{#1 JA_j}}
\def\ps{{\mathfrak{p}}}
\def\Piz{\varPi_0}
\def\Pis{\varPi_{\ps}}
\def\bN{{\bf N}}

\title{Freezing of energy of a soliton in an external potential}

\author{D. Bambusi\footnote{Dipartimento di Matematica, Universit\`a degli Studi di Milano, Via Saldini 50, I-20133
Milano. \newline
 \textit{Email: } \texttt{dario.bambusi@unimi.it}},
 A. Maspero\footnote{Dipartimento di Matematica, Universit\`a degli Studi di Milano, Via Saldini 50, I-20133
Milano. \newline
 \textit{Email: } \texttt{alberto.maspero@unimi.it}}}
\def\bv{{\bf v}}

\begin{document}

\maketitle

\begin{abstract}
In this paper we study the dynamics of a soliton in the generalized
NLS with a small external potential $\epsilon V$ of Schwartz class.
We prove that there exists an effective mechanical system describing
the dynamics of the soliton and that, for any positive integer $r$,
the energy of such a mechanical system is almost conserved up to times
of order $\epsilon^{-r}$. In the rotational invariant case we deduce
that the true orbit of the soliton remains close to the mechanical one
up to times of order $\epsilon^{-r}$.
\end{abstract}

\section{Introduction and Statement of the Main Result }

\subsection{Introduction}
Consider the equation
\begin{equation}
 \label{NLSV}
 \im \dt \psi = -\lap \psi  - \beta'(|\psi|^2)\psi + \epsilon V( x)
 \psi \ ,\quad x\in\R^3\ ,
\end{equation}
where $V$ is a potential of Schwartz class, $\beta\in
C^{\infty}(\R,\R)$ is a function fulfilling
\begin{equation}
\label{num}
\mmod{\beta^{(k)}(u)} \leq C_k \la u \ra^{1+p-k} \ , \quad \beta'(0) = 0 \, \quad p<2 \ ,
\end{equation}
and $\epsilon$ is a small parameter.

In the case $\epsilon=0$, under suitable assumptions on $\beta$,
equation \eqref{NLSV} admits solitary wave solutions, namely solutions
which travel with uniform velocity (solitons, for short). Such
solutions form an 8 dimensional {\it soliton manifold} $\Tr$ (see
\eqref{soliton.manifold} for a precise definition) parametrized by the
mass $m$ of the soliton, by its linear momentum $\bp$, by a Gauge
angle $q^4$ and by the barycentre $\bq$.

Take now $\epsilon\not=0$, then, up to higher order corrections, the
restriction of the Hamiltonian \eqref{NLSV} to the soliton manifold
$\Tr$ takes the form of an $m$ dependent constant plus
\begin{equation}
\label{hamsol}
H^\epsilon_{mech}(\bp,\bq)=\frac{\left|\bp\right|^2}{2m}+\epsilon
V^{eff}_m(\bq)\ ,
\end{equation}
where $V^{eff}_m$ is an effective potential (see \eqref{veff}), which
for large mass $m$ is close to $V$ (see e.g. \cite{frohlich1}).
Formally \eqref{hamsol} is the Hamiltonian of a particle subject to
the force due to the effective potential. However, the soliton
manifold is not invariant under the dynamics: the soliton and the rest
of the field are coupled, so the soliton is expected to move according
to the Hamilton equations of \eqref{hamsol} only approximately. In
particular the coupling is expected to lead to radiation of energy and
to an effective dissipation on the dynamics of the soliton.

The main result of the present paper is that the coupling between the
soliton and the rest of the field is not effective up to very long
times.  Precisely, if the initial datum $\psi_0$ is
$\cO(\epsilon^{1/2})$-close to the soliton manifold and the mechanical
energy $H^\epsilon_{mech}$ is $\cO(\epsilon)$, then one has
\begin{equation}
\label{estim}
|H^\epsilon_{mech}(t)-H^\epsilon_{mech}(0)|< C\epsilon^{3/2}\ ,\quad
|t|\leq \epsilon^{-r}\ , \quad \forall r\in\N \ ,
\end{equation}
for $\epsilon$ small enough.  

A particularly interesting corollary can be deduced if the initial
datum and the potential $V$ are axially symmetric. Indeed in such a
case the soliton's motion is one dimensional, and if its initial
energy is not a critical value of $\epsilon V^{eff}_m$, then the orbit
is $\cO(\epsilon^{3/2})$-close to the orbit of the mechanical
system. This is true for times of order $\epsilon^{-r}$, $\forall
r$. The most interesting case is the one in which the orbit of the
system \eqref{hamsol} is periodic: in such a case the true motion of
the soliton is also approximately periodic for very long times. Of
course the approximate period of the true motion is different from the
period of the orbit of \eqref{hamsol}.

\vskip 10pt

The problem of the dynamics of a soliton in an external potential has
been widely studied, and the results obtained so far can be
essentially divided into 2 groups: in the first group of papers, the
authors describe the dynamics of the soliton up to long, but finite
times
\cite{frohlich1,holmer_zworski,jonsson2006,salem2009,holmer2007,holmer2011},
while in the second group of papers the authors exploit dispersive
properties of the equations (in the case of potentials going to zero
at infinity) in order to study the asymptotic behaviour of the
soliton \cite{gustafson_nakanishi_tsai,gang_sigal2005,gang_sigal,
  gang_weinstein,deift_park,cuccagna_maeda2, cuccagna_maeda1}.

The results of the papers of the first group deal mainly with the case
of potentials of the form $V(\epsilon x)$ (no $\epsilon$ in front of
the potential) and in the most favorable cases (in particular when the
potential is confining) they show that, up to a small error and for a
time scale of order $\epsilon^{-3/2}$, the variables $(\bp,\bq)$
evolve according to the equations of the effective Hamiltonian
\eqref{hamsol} (see \cite{jonsson2006}). Some numerical computations
done in the case of localized potentials show that the true motions of
the soliton are actually different from the mechanical ones and that
the difference becomes macroscopic after a quite short time scale (see
\cite{holmer_zworski}).  We point
out that this is not surprising, since even in the case of classical
integrable finite dimensional systems, motions starting nearby get far
away after quite short time scales.

The classical way to get control of the dynamics for longer times
consists in renouncing to control the evolution of all the coordinates
and to keep control only on some relevant quantities, e.g. the actions
or the energy of some subsystem. In the case of the soliton's dynamics
in NLS, this is possible since the system turns out to be composed by
two subsystems whose evolution occurs over different time scales: the
time scale of the soliton's dynamics is of order $\epsilon^{-1}$,
while the time scale of the field is of order 1. The situation is
analogous to that met in the classical problem of realization of
holonomic constraints (see \cite{BGG1,BGG2,BG93,BGPP}), from which we
borrow ideas and techniques.

Coming to the results of the second group we first recall
\cite{gang_sigal}, in which the authors consider a potential of
the form $V(\epsilon x)$ with a nondegenerate minimum at $x=0$ and
prove that, for $\epsilon$ small, the solution with the
soliton at rest at the bottom of the well is asymptotically stable.

Our result pertains mainly initial data which are not close to the
minimum, and prove that the soliton dynamics is conservative up to
very long times, so that, up to such times it does not display phenomena of
asymptotic stability. In the axially symmetric case we conclude that
the soliton's orbit remains close to a mechanical orbit for the times
we are considering. We remark that the result of \cite{gang_sigal} is
obtained exactly under the assumptions of axially symmetric potential
and initial data in which we control the orbit of the soliton in the
present paper.

We recall now the result of \cite{deift_park} in which
the following equation is considered
$$
i \psi_t= -\psi_{xx}-q\delta_0 (x) \psi-\left|\psi\right|^2\psi \ ,
$$ with even initial data; here $\delta_0(x)$ is a Dirac delta
function playing the role of a potential. The authors exploit the fact
that such an equation is equivalent to an integrable system on the
half line and they describe the long time asymptotics of the dynamics,
in particular they show that the solution converges to a soliton at
rest at the origin plus radiation.

Finally we come to \cite{cuccagna_maeda2,cuccagna_maeda1}. The
authors consider equation \eqref{NLSV} with a soliton having positive
energy of order 1 and prove that, for $\epsilon$ small, the soliton
asymptotically behaves as a solution of the free NLS.  They also
consider a case in which $\epsilon=1$, but in this case they either
assume that
the soliton is far from the region where the potential is
significantly different from zero or has a large velocity.

\subsection{Main result}\label{mainres}

Equation \eqref{NLSV} is Hamiltonian with Hamiltonian function given
by
\begin{align}
\label{H}
H(\psi)& := H_0(\psi) + H_P(\psi)+\epsilon H_V(\psi)  \ ,
\\
 H_0(\psi)&:= \int_{\R^3} \mmod{\nabla \psi(x)}^2 \, d^3x \ , \qquad
  H_P(\psi):=  - \int_{\R^3} \beta(|\psi(x)|^2) \, d^3 x \ .
\\
H_V(\psi)&:= \int_{\R^3}  V(  x) |\psi(x)|^2 \, d^3x \ ,
\end{align}
To start with we will study the system in the energy space $H^1$
endowed by the scalar product and the symplectic form
\begin{equation}
\label{scalsympNLS}
\la \psi_1, \psi_2 \ra := 2 \mbox{ Re } \int_{\R^3} \psi_1(x)\,
\overline{\psi_2(x)} \, dx \ ; \quad \omega(\psi_1, \psi_2):= \la E
\psi_1, \psi_2 \ra \ ,
\end{equation}
where we denoted by $E:= \im$ the symplectic operator. In the
following we will denote by $J:= E^{-1} \equiv -\im$ the standard
Poisson tensor. For a $C^1$ function $H:H^{1} \to \R$ we denote by
$\nabla H$ its gradient, defined by the equation
$$ \di H(\psi)\Phi = \langle \nabla H, \Phi \rangle, \quad \forall
\Phi \in H^{1} \ .
$$
The Hamiltonian vector field $X_H$ of a Hamiltonian function $H$ is thus  given by
$X_{H}:= J\nabla H$ and the corresponding Hamiltonian equations are given by $\dot{\psi} = J \nabla H$.

The Hamiltonian \eqref{H} is invariant under the Gauge
transformation and, in the case $\epsilon=0$, it is also invariant
under translations. We denote by $\cP_j(\psi)$, $j=1,...,4$, the
corresponding conserved quantities, which are explicitly given by
\begin{align}
\label{momenta}
\cP_j(\psi) &:= \int_{\R^3} \overline{\psi(x)}\, \im \partial_{x_j}
\psi(x) \, dx \equiv \frac{1}{2} \la A_j \psi, \psi \ra \ , \qquad
j=1,2,3\ , 
\\
\label{mass1}
\cP_4(\psi) &:= \int_{\R^3} \mmod{\psi(x)}^2 \, dx  \equiv \frac{1}{2} \la A_4 \psi, \psi \ra \ ,
\end{align}
where $A_j:= \im \partial_{x_j}$, $1 \leq j \leq 3$ and $A_4 := \uno$.
The Hamiltonian flows of the $\cP_j$'s will be denoted by
\begin{align}
\label{trans.boost}
\left[ e^{{ q}  J A_j}\psi\right](x):= \psi(x - {q}{\bf e}_j)
\ ,\ j=1,2,3\ ,
\\
\nonumber
e^{q J A_4}\psi:= e^{-i q} \psi \ ,
\end{align}
and of course they are the symmetries of the Hamiltonian when
$\epsilon=0$. Here we denoted $\bee_1:=(1,0,0)$ and similarly
$\bee_2$, $\bee_3$.

We recall that the solitons are the critical points of
$H|_{\epsilon=0}$ at fixed values of the momenta $\cP_j$. They can be
constructed starting from the ground state $b_{\cE}$ with zero
velocity, which is the minimum of $H|_{\epsilon=0}$ constrained to a
surface of constant $\cP_4$. In order to ensure existence of the
ground state we assume:

\begin{enumerate}[(H1)]
 \item There exists an open interval $\mathcal{I} \subset \R$ such that,  $\forall \, \cE \in \mathcal{I}$, the equation
\begin{equation}
 \label{solitonNLS}
 -\lap b_{\cE} -  \beta'(b^2_{\cE})b_{\cE} + \cE b_{\cE}=0 \
\end{equation}
admits a $C^{\infty}$ family of real, positive, radially symmetric functions
$b_{\cE}$ belonging to the Schwartz space.

\end{enumerate}

The quantity
\begin{equation}
\label{mass}
m=m(\cE):=\cP_4(b_{\cE})/2
\end{equation}
will play the role of mass of the soliton. Defining
$$\widetilde\eta(\bv,\cE):=e^{- \im \frac{\bv\cdot x}{2}}b_{\cE}$$
one gets the initial datum for a soliton
moving with velocity $\bv\equiv(v_1,v_2,v_3)$.

We also assume that
\begin{itemize}
\item[(H2)] One has $\frac{d}{d \cE}\norm{b_{\cE}}^2_{L^2}>0 $, $\forall \,
  \cE \in \mathcal{I}$ ,
\end{itemize}
so that $b$ can be parametrized by the mass $m$ instead of $\cE$. An
explicit computation gives
$$
\cP_j(\widetilde\eta(\bv,\cE) )= m v_j\ ,\quad j=1,2,3\ ,
$$ which shows the analogy with the momentum of a particle. In order
to state our main theorem it is useful to consider $m$ as a
parameter and to take into account also the translations of the states
$\widetilde\eta$. We will denote
\begin{align}
\label{eta1}
\eta_m(\bp,\bq):=e^{\sum_{j=1}^3q^jJA_j} \ \widetilde\eta (\bp/m,\cE(m))\ ,
\\
\label{we}
\bp:=(p_1,p_2,p_3)\ ,\quad \bq:=(q^1,q^2,q^3)\ .
\end{align}

\begin{remark}
\label{rem1}
Fix an initial value $(\bp_0,\bq_0)$ for momentum and position, then
the solution of \eqref{NLSV} with $\epsilon=0$, corresponding to the
initial datum $\eqref{eta1}$ has the form
\begin{equation}
\label{sol0}
\psi(x,t)=e^{-\im
  (\cE+\frac{|v|^2}{4})t}\null\hskip2pt\eta_m(\bp_0,\bq_0+\frac{\bp_0}{m}
t) \ .
\end{equation}
Consider the linearization of eq. \eqref{NLSV}, with $\epsilon=0$, at such
solution: in terms of real and imaginary parts of $\psi$ it can be
written in the form $\dot \psi=L_0\psi$ with
\begin{equation}
\label{L0}
L_0:= \left[
\begin{matrix}
0& -L_-
\\
L_+&0
\end{matrix}
\right]\ ,
\end{equation}
and
\begin{equation}
 L_+:= -\lap + \cE -\beta'(b^2_{\cE}) \ , \qquad L_{-}:=-\lap + \cE -\beta'(b^2_{\cE}) - 2\beta^{''}(b^2_{\cE})b^2_{\cE} \ .
\end{equation}
It is classical that (due to the symmetries of the system) $0$ is an
eigenvalue of $L_0$ with multiplicity at least $8$. Furthermore $L_0$
has a purely imaginary continuous spectrum given by
$\bigcup_{\pm}\pm\im[\cE,+\infty)$.
\end{remark}

We assume
\begin{itemize}
\item[(H3)]  The Kernel of the operator $L_+$ is generated by $b_{\cE}$ and
the Kernel of the operator $L_-$ is generated $\partial _jb_\cE$,
$j=1,2,3$. 
\item[(H4)] $\pm i\cE$ are not resonances of $L_0$.
\item[(H5)] The pure point spectrum of $L_0$ contains only zero.
\end{itemize}

\begin{remark}
\label{exis}
Under assumptions (H2,H3) above, the solutions \eqref{sol0} are
orbitally stable when $\epsilon=0$ and, under (H4) and the further
assumption that the so called nonlinear Fermi Golden Rule holds they
are also asymptotically stable (see \cite{buslaev_perelman,cuccagna_mizumachi, bambusi2013, cuccagna2014}).
\end{remark}
\begin{remark}
\label{h5}
Assumption (H5) is here required only for the sake of simplicity: we
expect that using the methods of \cite{bambusi_cuccagna,bambusi2013}
(see also \cite{cuccagna_maeda2,cuccagna_maeda1}) one can remove
such an assumption.
\end{remark}

In order to state the main theorem we consider the mechanical
Hamiltonian system \eqref{hamsol} with
\begin{equation}
\label{veff}
V^{eff}_m(\bq):=\intre V(x+\bq) b^2_{\cE(m)}(x) \, d^3x\ .
\end{equation}

Our main result is the following theorem.

\begin{theorem}
\label{main}
Fix a positive integer $r \in \N$ and positive constants
$K_1,K_2,T_0$. Then there exist positive constants $\epsilon_r$, $C_1$
s.t. if $0\leq\epsilon<\epsilon_r$, then the following holds true:
consider an initial datum $\psi_0\in H^1$ s.t. there exist $(\bar
m,\bar \alpha)$ and $(\bar \bp,\bar\bq)$ with
\begin{equation}
\begin{aligned}
\label{epsilon0}
   \norm{\psi_0 - e^{\im \bar \alpha} \eta_{\bar m }(\bar\bp,\bar\bq)}_{H^1} \leq
   K_1\epsilon^{1/2} \\ H^\epsilon_{mech}(\bar\bp,\bar\bq)<K_2 \epsilon\ ,
\end{aligned}
\end{equation}
then, for $|t|\leq T_0 \epsilon^{-r}$,
the solution $\psi(t)$ of \eqref{NLSV} exists in $H^1$ and admits
the decomposition
\begin{equation}
\label{dec1}
\psi(t) :=e^{\im\alpha(t)} \eta_{m}(\bp(t),\bq(t)) + \phi(t)   \ ,
\end{equation} with a constant $m$ and smooth functions
$\bp(t),\bq(t), \alpha(t)$ s.t.
\begin{align}
\label{sti1}
\left|H^\epsilon_{mech}(\bp(t),\bq(t))-
H^\epsilon_{mech}(\bp(0),\bq(0))\right|\leq
C_1\epsilon^{3/2}\ ,\quad|t|\leq\frac{T_0}{\epsilon^r}\ .
\end{align}
 Furthermore, for the same times one has
\begin{equation}
\label{stif}
\norma{\phi(t)}_{H^1}\leq C_1\epsilon^{1/2}\ .
\end{equation}
\end{theorem}

\begin{remark}
\label{correz}
In the above statement, the quantities $\bar m,\bar \alpha, \bar\bp,
\bar\bq,$ do not coincide with $m$ and with the initial values of
$\alpha(t),\bp(t),\bq(t)$. This is due to the fact that $\psi_0-e^{\im
  \bar \alpha} \eta_{\bar m }(\bar\bp,\bar\bq)$ could have some
``nontrivial component along the soliton manifold'', so one has to add
a small correction to $(\bar m,\bar\alpha,\bar\bp,\bar\bq)$ in order to
avoid this. To give a precise meaning to the above loose statement is
non trivial and we will show in Section \ref{ada} how this has to be
done.
\end{remark}

As anticipated above, in the axially symmetric case one can deduce a
particularly interesting corollary. To come to its statement consider
the case where the potential is symmetric under rotations about the
$z$ axis (of course the choice of such an axis is arbitrary) and take
an initial datum symmetric under rotations about the same axis, assume
it fulfills \eqref{epsilon0}, then, from the proof, one has that the functions
$(\bp(t),\bq(t))$ also belong to the $z$-axis.  Consider the solution
$(\bp_m(t),\bq_m(t))$ of the Hamiltonian system $H^\epsilon_{mech}$
with initial data $(\bp(0),\bq(0))$. Denote by
$$
\Gamma_m:=\bigcup_{t\in\R} \left\{ (\bp_m(t),\bq_m(t))\right\}
$$ the mechanical orbit (which of course is a level set of
$H^\epsilon_{mech}$ restricted to the $z$-axis), and introduce in
$\R^6$ the weighted norm
\begin{equation}
\label{weight}
\left\|(\bp,\bq)\right\|^2_{\epsilon}:=\sum_{k=1}^3(p_k^2+\epsilon
q_k^2)\ ,
\end{equation}
then one has the following corollary.
\begin{corollary}
\label{1d}
With the above notations, and under the assumptions of Theorem
\ref{main}, assume also that
$\frac{1}{\epsilon}H^\epsilon_{mech}(\bp(0),\bq(0))$ is not a critical
value of $V^{eff}_m$, then there exists a positive constant $C_2$ such that
the functions $(\bp(t),\bq(t))$ of equation \eqref{dec1} fulfill
\begin{equation}
\label{sti3}
d_\epsilon(( \bp(t),\bq(t) );\Gamma_m)\leq C_2 \epsilon^{3/2}\ , \qquad  \forall
\left|t\right|\leq T_0 \epsilon^{-r}\ ,
\end{equation}
where $d_\epsilon(.;.)$ is the distance in the norm \eqref{weight}.
\end{corollary}

Of course the most interesting case is the one in which $\Gamma_m$ is
a closed curve and thus the soliton essentially performs a periodic
orbit for the considered times.

\subsection{Scheme of the proof}\label{scpro}

The proof proceeds essentially in three steps: first we introduce a
system of coordinates $(p,q,\phi)$ close to the soliton manifold in
which the $p$'s are the momenta $\cP_j$ of a soliton, the $q$'s the
coordinates of the barycentre and a Gauge angle and $\phi$ represents
the free field (see eq. \eqref{p.q.phi.coord}). Such coordinates are
not canonical, so we have to prove a Darboux theorem in order to show
that it is possible to deform the coordinates into canonical
ones. This is obtained along the lines of the Darboux theorem
of \cite{bambusi2013} (see also \cite{cuccagna2014}).

Then we write the Hamiltonian in such canonical coordinates. After a
suitable rescaling of the variables, it turns out that $H$ has the
structure
\begin{equation}
\label{stuH}
H=\frac{1}{2}\left\langle
EL_0\phi;\phi\right\rangle+\epsilon^{1/2}
\left[\frac{|\bp|^2}{2m}+V^{eff}_m(\bq) \,\right]+\cO(\epsilon) \ ,
\end{equation}
where $\phi$ belongs to the spectral subspace corresponding to the
continuous spectrum of $L_0$ (defined by \eqref{L0}). Since
$\sigma_c(L_0)=\pm\im[\cE,+\infty)$, this implies that the typical
  frequency of the field is of order 1, while the typical frequency of
  the soliton is of order $\epsilon^{1/2}$ so that we are in the same
  framework met in the problem of realization of holonomic constraints
  in classical mechanics \cite{BGG1,BGG2,BG93,BGPP}. The classical
  methods used in that context consist in developing a normal form
  theory in which one eliminates from the Hamiltonian all the terms
  which are nonresonant with respect to the frequencies of the fast
  system, the field $\phi$, in our case. In classical mechanics, this
  is possible up to a remainder which is of arbitrary order or
  exponentially small in $\epsilon$. However, in the present case this
  is impossible since the spectrum of $L_0$ is continuous. So the only
  thing we can do and we actually do, is to remove from the
  Hamiltonian all the terms which are linear in the field $\phi$. This
  is the second step of the proof.

As a third and final step we exploit the so obtained normal form in order to get a control
of the dynamics. The main step in order to do that consists in showing
that the free field $\phi$ fulfills Strichartz estimates (as in the
linear NLS) and to exploit the Hamiltonian structure in order to
deduce that $H_L(\phi):= \frac{1}{2}\left\langle
EL_0\phi;\phi\right\rangle $ changes at most by $\cO(\epsilon^{3/2})$ up
to times of order $\epsilon^{-r}$. To this end we use some Strichartz
type estimates for time dependent linear operators which were already
obtained in \cite{beceanu,bambusi2013,perelman}. Finally we exploit
conservation of the Hamiltonian in order to conclude the proof.

The rest of the paper is organized as follows. In Sect. \ref{ada} we
introduce Darboux coordinates close to the soliton manifold. In that
section we also introduce some classes of maps that will play an
important role in the paper and that substitute standard smooth
maps. In Sect.~\ref{coordiHam} we rewrite the Hamiltonian in the
Darboux coordinates. Actually, the Hamiltonian has the same form also
after applying any change of coordinates belonging to a suitable class
of maps, which in particular will be the one used to put the system in
normal form. In Sect.~\ref{section_normal_form} we construct the
transformation putting the system in normal form. In
Sect. \ref{disphi} we use the normal form and dispersive properties of
NLS in order to get estimates of the solution and the proof of the
main theorem. Finally, in the appendixes we prove two auxiliary
results.

\vskip10pt
\noindent
{\it Acknowledgements.} This research was founded by the Prin project
2010-2011 ``Teorie geometriche e analitiche dei sistemi Hamiltoniani
in dimensioni finite e infinite''. The second author has been
partially supported by the Swiss National Science Foundation.

\section{Adapted coordinates and the Darboux theorem}\label{ada}

In the course of the paper we will need the standard Lebesgue spaces
$L^p$, the standard Sobolev spaces $W^{s,p}$ of  functions whose
weak derivatives of order $s$ are of class $L^p$, and the
corresponding Hilbert spaces $H^s:=W^{s,2}$. Furthermore, we need the
scale of Hilbert spaces $\cH^{s,k}(\R^3, \C) $ defined by
$$
\cH^{s,k}(\R^3, \C) := \{ \psi \mbox{ s.t. } \norm{\psi}_{\cH^{s,k}}:= \norm{\jap{x}^k (-\Delta +1)^{s/2} \psi}_{L^2(\R^3, \C)} < \infty \} \ ,
$$ and we will denote $\cH^{\infty} := \bigcap_{s,k} \cH^{s,k}$,
$\cH^{-\infty} := \bigcup_{s,k} \cH^{s,k}$.

Finally we will use the notation $a\sleq b$ to mean ``there exists a
constant $C$, independent of all the relevant parameters, such that
$a\leq Cb$''.

It is convenient to change slightly the notation concerning the
soliton: first we fix once for all a positive value of the mass
corresponding to which a ground state $b_{\cE(m)}$ exists.  We will
work close to the manifold
\begin{equation}
\label{t0}
\Tr_0:=\bigcup_{q\in\R^4}e^{q^jJA_j}b_{\cE(m)}\ ,\quad
q\equiv(q_1,q_2,q_3,q_4) \ ,
\end{equation}
where sum over repeated indexes is understood.

{\bf From now on we will denote by $\eta_p$ the following ground
  state}:
\begin{equation}
\label{grp}
\etap:= e^{- \im  \sum_{k=1}^{3}\frac{p_k}{2(m+p_4/2)} x_k} b_{\cE(m+p_4/2)} \ ,
\end{equation}
so that $\eta_0=b_{\cE(m)}$. The ground state $\eta_p$ fulfills the
equation
\begin{equation}
\label{gro}
-\Delta \eta_p+\nabla H_P(\eta_p)-\lambda^j(p)A_j\eta_p=0\ ,
\end{equation}
with
\begin{equation}
\label{lambda}
\lambda^j(p):= \frac{ p_j}{m+p_4/2}, \ \ \ j=1,2,3 \ , \quad
\lambda^4(p):= - \left( \cE(m+p_4/2) +\frac{|p|^2}{4(m+p_4/2)^2} \right)
\ .
\end{equation}
Furthermore, one has that
\begin{align*}
\cP_j(\eta_p) = p_j \ , \ j=1,2,3, \qquad \cP_4(\eta_p)= 2m + p_4 \ .
\end{align*}
Having fixed a small neighbourhood $\cJ\subset \R^4$ of $0$ we define
the soliton manifold by
\begin{equation}
\label{soliton.manifold}
 \Tr:= \bigcup_{q \in \R^4, \, p  \in \cJ} e^{q^j J A_j} \etap \ .
 \end{equation}
The tangent space to $\Tr$ at the point $\etap$ is generated by
$$ T_{\etap} \Tr:= \mbox{span} \left\{JA_j \etap, \etapk \right\}  $$
while  its symplectic orthogonal space $T^{\angle}_{\etap}\Tr $ is given by
\begin{equation}
\begin{aligned}
\label{symplectic.ort}
T^{\angle}_{\etap}\Tr= \Big\{ \Psi \in \cH^{-\infty}:\, & \  \omega(JA_j\etap, \Psi)= \la A_j\etap, \Psi \ra =0, \qquad \omega\left(\etapk, \Psi\right)= \langle E \etapk , \Psi \rangle =0
 \Big\} \ .
\end{aligned}
\end{equation}
We decompose the space $\cH^{-\infty}$ in the direct sum of the tangent
space of $\Tr$ at $\etap$ and its symplectic orthogonal. More
precisely we have the following lemma, whose proof is obtained by
taking the scalar product of \eqref{2.13} with $JA_j\eta_p$ or
with $\etapk$.
\begin{lemma}
 One has $\cH^{-\infty} = T_{\etap} \Tr \oplus T^{\angle}_{\etap}\Tr$. Explicitly the  decomposition of a vector $\Psi \in \cH^{-\infty}$ is given by
\begin{equation}\label{2.13}
 \Psi = P_j \etapj + Q^j JA_j \etap +\varPhi_{p}
\end{equation}
with
\begin{equation}
 P_j = \la A_j \etap, \Psi \ra \  , \quad Q^j =  -  \la E \etapj, \Psi \ra
\end{equation}
and $\varPhi_{p} \in T^{\angle}_{\etap}\Tr$ given by
\begin{equation}
\label{ortvect}
 \varPhi_{p}= \Psi - \la A_j \etap, \Psi \ra \etapj + \la E \etapj, \Psi \ra JA_j \etap \ .
\end{equation}
\end{lemma}

\begin{remark}
\label{pi}
A key object for the whole theory we will develop is the projector
$\Pip $ on $T^{\angle}_{\eta_p}\Tr$ defined by
\begin{equation}
\label{pip}
\Pip \Psi:=\Psi - \la A_j \etap, \Psi \ra \etapj + \la E \etapj, \Psi \ra JA_j
\etap\ .
\end{equation}
Its most important property is that it is a smoothing
perturbation of the identity, namely $\id-\Pip$ maps smoothly
$\cH^{-s_1, -k_1}$ into $\cH^{s_2, k_2}$ for every $s_1, k_1, s_2, k_2 \in \R$.
\end{remark}

The reference space that we will  use in order to parametrize the free
field will be
\begin{equation}
\cV^{s,k}:= \Pipz \cH^{s,k}
\end{equation}
which we endow by the topology of $\cH^{s,k}$.

In order to describe a neighborhood of $\Tr_0$ we will use coordinates
\begin{equation}
\label{ccorK}
(p,q,\phi)\in \cK^{s,k}:= \R^4\times\R^4\times \cV^{s,k}\ ,
\end{equation}
s.t. $\Tr_0$ coincides with $p=\phi=0$ (actually $q$ varies in
$\R^3\times\T$, but we work in the covering space $\R^4$).
We endow the scale $\cK\equiv \{\cK^{s,k}\}$ with the norm
$$
\norm{(p,q,\phi)}_{\cK^{s,k}}^2 := \norm{p}_{\R^4}^2 + \norm{q}_{\R^4}^2 + \norm{\phi}_{\cV^{s,k}}^2 \ .
$$
 By abuse of language, when
dealing with the scale $\cK$, we will mention
$\Tr_0$ in order to mean the manifold $p=\phi=0$.

The coordinates we will use are defined by the map
\begin{equation}
\label{p.q.phi.coord}
 \coS(p, q, \phi):= e^{q^j J A_j}
 \left( \eta_p +\Pip\phi \right)\ .
\end{equation}
The map $\coS$ does not depend smoothly on $q$
(due to the unboundedness of $JA_j$, $j=1,2,3$) and this will be the
source of all the difficulties. Nevertheless we have the following
lemma.

\begin{lemma}
\label{invcoo}
In a neighbourhood of $\Tr_0$ there exists a unique inverse map
$\coS^{-1}$ of $\coS$, with the following properties: denote
$(p(\psi),q(\psi),\phi(\psi)):=\coS^{-1}(\psi)$, then $\forall r,s$
there exists an open neighbourhood $\cU_{r,s}\subset \cH^{r,s}$ of
$\Tr_0$ s.t.
$$
\cU_{r,s}\ni\psi\mapsto (p(\psi),q(\psi))\in\R^8
$$
is $C^{\infty}$; the map
$$
\cU_{r,s}\ni\psi\mapsto \phi(\psi) \in \cV^{r,s}
$$
is continuous and maps bounded sets in bounded sets.
\end{lemma}
The proof of this lemma, which is a small variant of Lemma 22 of
\cite{bambusi2013} is reported in appendix for the sake of completeness.

\begin{corollary}
\label{bar}
If $\psi\in H^1$ is s.t.
\begin{equation}
\label{bar1}
\norma{\psi-e^{\bar q^jJA_j}\eta_{\bar p}}_{H^1}\leq K_1 \sqrt \epsilon\ ,
\end{equation}
for some $(\bar p,\bar q)\in \R^8$, then there exist $(p,q,\phi)$ such
that  $\psi=\coS(p,q,\phi)$
and
\begin{equation}
\label{bar2}
\left\| p-\bar p\right\|\sleq  K_1 \sqrt \epsilon\ ,\quad \left\|
q-\bar q\right\|\sleq K_1 \sqrt \epsilon\ ,\quad
\left\|\phi\right\|_{\cV^{1,0} }\sleq K_1\sqrt\epsilon\ .
\end{equation}
\end{corollary}

We introduce now (following \cite{bambusi2013}) some classes of maps
that will be used in all the rest of the paper. In the corresponding
definitions we will use different scales of Hilbert
spaces. Essentially, besides the ones already introduced we will use
the trivial one composed by one space, namely $\R^m$ or $\C^m$ or the
scale $\wcK:=\R^4\times\cK$, in which the first factor $\R^4$ is the
space in which varies a 4 dimensional vector $N=(N_j)$ that will
eventually be set equal to $\cP(\phi)\equiv (\cP_j(\phi))$. This is
needed since we will meet functions which depend in a smoothing way on
$\phi$ except for the special dependence through the functions
$\cP_j$. Finally, we will also consider scales with one additional
component, this is needed in order to add a small parameter.

\begin{definition}
Given two scales of Hilbert spaces $\cH\equiv \{\cH^{s_1,k_1} \}$ and
$\wtcH\equiv \{\wtcH^{s_1,k_1} \}$, we will say that a map
$f$ is of class $\ALS(\cH, \wtcH)$ if $\forall r, s_2, k_2 \geq 0$
there exists $s_1, k_1$ and an open neighborhood
$\cU_{rs_1k_1} \subset \cH^{s_1,k_1}$ of $\Tr_0$, such that
\begin{equation}
f \in C^r\left(\cU_{rs_1k_1}, \widetilde{\cH}^{s_2,k_2} \right) \ .
\end{equation}
Such maps will be called {\em almost smooth}.
\end{definition}

\begin{definition}
A map $f$ will be said to be {\em regularizing} or {\em of class
  $C_R(\wcK , \cK )$} if $\forall r, s_1, k_1, s_2, k_2 \geq 0$ there
exists an open neighbourhood $\cU_{rs_1k_1s_2k_2} \subset
\wcK^{-s_1,-k_1}$ of $\Tr_0$, such that
\begin{equation}
f \in C^r\left(\cU_{rs_1 k_1s_2k_2}, \cK^{s_2,k_2} \right) \ .
\end{equation}
\end{definition}

\begin{definition}
\label{smoothing}
For $i, j \geq 0$, a map $S$ will be said to be of class $\SH^{i}_{j}$
if there exists a regularizing map $\widetilde S \in C_R(\wcK, \cH)$,
with the property that $S(p, q, \phi) = \wtS(\cP(\phi),p, q, \phi)$
and such that, $\forall s_1, k_1, s_2, k_2 \geq 0 $ there exists
$C\geq 0$
s.t.
\begin{equation}
 \label{smoothingequation}
\normk{\widetilde S(N,p, q, \phi)}{s_2, k_2} \leq C\left(\sum_{l_1+l_2=i} |p|^{l_1}|N|^{l_2}\right) \norm{\phi}_{\cV^{-s_1, -k_1}}^{j}
\end{equation}
$\forall (N,p, q, \phi)$ in some neighborhood of $\{0\}\times\Tr_0$.

Functions belonging to the classes $\SH^{i}_{j}$ will be called {\em
  smoothing}.
\end{definition}

We will often consider the case of smoothing maps taking values in
$\R^n$ or $\C^n$. In this case we will denote the corresponding
classes with the special notation $\SR^i_j$. In the following we will
identify a smoothing function $S$ with the corresponding function
$\widetilde S$. We will also consider the case of time dependent
smoothing maps, in which the dependence on time is also assumed to be
smooth.

\begin{remark}
\label{notation}
In what follows the specific form of smoothing functions in the
classes $\cS^{i}_j$ or $\cR^i_j$ is not important, for this reason we
will almost always denote such functions simply by $S^i_j$ or $R^i_j$,
and the same letter will denote different objects.  For example we
will meet equalities of the form
 $$
 S^1_1 + S^1_2 = S^1_1 \ ,
 $$ where obviously the function $S^1_1$ at r.h.s. is different from
 that at l.h.s.
\end{remark}

\begin{remark}
\label{pishi}
By explicit computation one has that, if $s^j\in \cR^{k}_l$ then
$$ \Pipz\ex{s^j}\phi=\ex{s^j}( \phi+S^k_{l+1})\ .
$$
\end{remark}

The last class of maps that we will need is the following one
\begin{definition}
\label{almost}
A map $\cA$ is said to be an almost smoothing perturbation of the
identity of class $\ap^k_{l,i}$ if there exist smoothing functions
$\alpha,P,Q\in \SR^{k}_l$ for some $l,k\geq 0$ and $S_i^k\in\cS^k_i$
for some $i \geq 0$, s.t.
\begin{equation}
\label{alm}
\cA(N, p, q, \phi)=\Big(p + P(N, p, q, \phi),q + Q(N, p, q, \phi),\Pipz
e^{\alpha^j(N, p, q, \phi)JA_j}(\phi+S^k_i(N, p, q, \phi))\Big)\ .
\end{equation}
\end{definition}
\begin{remark}
\label{pj}
If $\cA\in\ap^k_{l,i}$ is an almost smoothing perturbation of the identity, then
one has
$$
\cP_j(\cA(N,p,q,\phi))=\cP_j(\phi)+R^k_{i+1}+ R^{2k}_{2i}\ .
$$
\end{remark}

Almost smoothing perturbations of the identity appear mainly as flows
of Hamiltonian vector fields of smoothing Hamiltonians. Precisely one
has the following lemma.
\begin{lemma}
\label{l.6}
Let $s^l,P,Q\in \resto^a_j$, $X\in\Sc^a_i $ $j\geq i\geq 1$, $a\geq 1$
be smoothing functions, and consider the system
\begin{equation}
\label{l.5.12}
\dot p=P(N, p, q, \phi)\ ,\quad \dot q =Q(N, p, q, \phi)\ ,\quad
\dot \phi= s^l(N, p, q, \phi)\Pipz JA_l\phi+X(N, p, q, \phi)\ .
\end{equation}
Then for $|t|\leq 1$, the corresponding flow $\cA_t$ exists in a
sufficiently small neighborhood of $\Tr_0$ in $\cK^{1,0}$, and for any
$|t|\leq 1$ it is an almost smoothing perturbation of the identity in the class $\ap^a_{j,i} $, namely it has the form
\begin{equation}
\label{flt}
\cA_t=\Big(p+\wtP(t,N, p, q, \phi),q+ \wtQ(t,N, p, q, \phi), \Pipz\e^{\alpha^l(t,N, p, q, \phi)JA_l}(\phi+S(t,N, p, q, \phi)) \Big)
\end{equation}
with $\wtP,\wtQ,\alpha^l\in\resto^a_j$ and
$S\in\Sc^a_{i}$.
\end{lemma}

In Appendix \ref{flow.chi} we will give the proof of Lemma
\ref{l.6.1} which is a small variant of the above lemma. Actually the
proof is a small variant of the proof of Lemma 3 of
\cite{bambusi2013}.

\begin{remark}
\label{invers}
Since $\cA_t $ is the flow of a vector field, the time $-t$ flow,
namely $\cA_{-t}$ is its inverse. Thus we have that, at least in this
case the inverse of the map \eqref{flt} exists and has the same
structure.
\end{remark}

\begin{remark}
\label{timedep}
Lemma \ref{l.6} holds also for time dependent vector fields with the
structure \eqref{l.5.12}. Also the so constructed almost smoothing
perturbations of the identity are invertible since the inverse is also
constructed as a flow.
\end{remark}

The coordinates \eqref{p.q.phi.coord} are not canonical. Let
$\Omega:=\coS^*\omega$ be the symplectic form in the variables
$(p,q,\phi)$ and define the reference symplectic form $\Omega_0$ by
\begin{equation}
\label{sym.ref}
\Omega_0\Big((P_1,Q_1,\Phi_1);(P_2,Q_2,\Phi_2)\Big)=
\sum_{j}\left(Q_{1j}P_{2j}-Q_{2j}P_{2j} \right)+ \langle
E  \Phi_1; \Phi_2\rangle\ ,
\end{equation}
then the following theorem holds.

\begin{theorem}
\label{darboux}
(Darboux theorem) There exists an almost smoothing perturbation of the
identity $\cD\in\ap^1_{0,1}$ of the form
\begin{equation}
\label{t.1.1}
\cD(p, q, \phi ) = \left( p -N + R^1_2, \ q + R^1_2, \ \Pipz e^{\alpha^j
  J A_j} (\phi + S^1_1) \right)
\end{equation}
with $\alpha \in \cR^{1}_2$, such that
$\cD^*\Omega=\Omega_0$.
Furthermore the maps $R^1_2,S^1_1,$ and $\alpha^j$ are independent of
the $q$ variables.

Finally $\cD $ is invertible and its inverse has the same structure.
\end{theorem}

\begin{remark}
In the Darboux coordinates  the Hamilton equations of a
Hamiltonian function $H$ have the form
\begin{equation}
\begin{cases}
& \dot p = - \frac{\partial H}{\partial q}(p, q, \phi) \\
& \dot q =  \frac{\partial H}{\partial  p}(p, q, \phi) \\
& \dot \phi = \Pipz J\nabla_{\bar \phi} H(p, q, \phi)
\end{cases} \ .\end{equation}
\end{remark}

\begin{remark}
\label{invar}
If a Hamiltonian function $H$ is invariant under the group action
$e^{q^jJA_j}$, namely $H(e^{q^jJA_j}\psi  )=H(\psi)$, then
in the Darboux coordinates just introduced it is independent of the
variables $q$. This is true since $H\circ \coS$ is independent of $q$
and the property is preserved by the coordinate change \eqref{flt}. In
particular this is true for the Hamiltonian \eqref{H} when
$\epsilon=0$, while, when $\epsilon\not=0$ the Hamiltonian is only
independent of $q^4$.

As a consequence when $\epsilon =0$, the  $p_j$'s, $1 \leq j \leq 4$, are integrals of motion for the Hamiltonian system, while when $\epsilon \neq 0$ only $p_4$ is an integral of motion.
\end{remark}

The proof of Theorem \ref{darboux}, which is a variant of the proof of
Theorem 3 of \cite{bambusi2013} will occupy the rest of the section.

\begin{remark}
    Since $\be$ is real valued and radial symmetric, $\eta_p$
        fulfills the orthogonality conditions
    \begin{equation}
    \label{isotropic.cond}
    \begin{aligned}
    &\la \etapk, E \etapj  \ra = \la J A_j \etap , A_k \etap \ra = 0 \qquad \forall \, 1 \leq j,k \leq 4 \ .
    \end{aligned}
    \end{equation}
\end{remark}

\begin{remark}
Recall the definition of $\Pip$ in \eqref{ortvect}. An explicit computation shows that the adjoint of $\Pip$ is given by
$$
\Pip^* \Psi = \Psi - \la \etapj, \Psi \ra A_j \etap + \la J A_j \etap, \Psi \ra E \etapj \ .
$$
\end{remark}
\begin{remark}
The following formulas will be useful in the following
\begin{equation}
\label{proj.formulae}
\begin{aligned}
E\Pip = \Pip^* E, \quad J \Pip = \Pip^* J, \quad \Pipj =\frac{\partial
  \Pip^2}{\partial p_j} = \Pip \Pipj + \Pipj \Pip , \\ \Pip \Pipj \Pip
=0 , \quad \left(\Pipj\right)^* = \frac{\partial \Pip^*}{\partial p_j}
, \quad E \Pipj= \frac{\partial \Pip^*}{\partial p_j} E \ .
\end{aligned}
\end{equation}
\end{remark}
\begin{remark}
    \label{smoothing.proj}
    For any  $s_1, k_1, s_2, k_2 \in \R$ one has
    \begin{equation}
\label{A.3}
    \normk{(\Pip-\Pipz)\phi}{s_2, k_2} \sleq  |p |
        \normk{\phi}{-s_1, - k_1} \ ,
    \end{equation}
    and by \eqref{proj.formulae} one has
    \begin{equation}
\label{A.3.2}
\normk{\Pip \Pipj \phi}{s_2, k_2} \sleq  |p | \normk{\phi}{-s_1, - k_1} \ .
    \end{equation}
\end{remark}
\begin{remark} \label{smoothing.proj2}
Consider $\Pip: \cV^{-\infty} \rightarrow \Pip \cV^{-\infty}$. It has
the structure $\Pip = \id + \left(\Pip - \Pipz \right)$ Thus, by
\eqref{A.3}, $\Pip$, as an operator on $\cV^{-\infty}$, is a smoothing
perturbation of the identity and it is invertible by Neumann
series. Furthermore its inverse $\widetilde{\Pip}^{-1}$ has the form
$\widetilde{\Pip}^{-1} = \id + S$ with $S$ fulfilling an estimate
equal to \eqref{A.3}.
\end{remark}

To begin with we compute the
symplectic form and a potential form for it in the coordinates
introduced by $\coS$, cf. eq. \eqref{p.q.phi.coord}.

\begin{lemma}
\label{l.1}
Define the 1-form $\Theta$ by
\begin{equation}
\label{l.1.1}
\Theta(P,Q,\Phi)= \frac{1}{2}\langle E\Pip\phi;\dep j\phi \rangle
P_{j}+ \left(-p_j+\frac{1}{2}\langle A_j\Pip\phi; \Pip
\phi\rangle)\right)Q^j +\langle E\Pip\phi;\Phi\rangle\
\end{equation}
(by this notation we mean that the r.h.s. gives the action of the form
$\Theta$ at the point $(p,q,\phi)$ on a vector $(P,Q,\Phi)$). Then one has $\di \Theta=\Omega$, and
therefore
\begin{equation}
\label{symplsoliton}
\begin{aligned}
&\Omega \left( (P_1, Q_1, \Phi_1); ((P_2, Q_2, \Phi_2))\right)= \la E
  \Pip \Phi_1, \Pip \Phi_2 \ra + Q_1^j P_2^j - P_1^j Q_2^j
\\
&+\la E
  \Pipj \phi, \Pipk \phi \ra P_1^j P_2^k + \frac{1}{2}
  \frac{\partial}{\partial p_k} \la A_j \Pip \phi, \Pip \phi \ra Q_j^1
  P_2^k - \frac{1}{2} \frac{\partial}{\partial p_j} \la A_k \Pip \phi,
  \Pip \phi \ra P_j^1 Q_2^k
\\
&
+ \la E\Pipj \phi, \Pip \Phi_2 \ra
  P_1^j - \la E\Pipk \phi, \Pip \Phi_1 \ra P_2^k
\\
& + \la A_j \Pip
  \phi, \Pip \Phi_2 \ra Q_1^j - \la A_k \Pip \phi, \Pip \Phi_1 \ra
  Q_2^k \ .
\end{aligned}
\end{equation}
\end{lemma}
\begin{proof} We compute $\coS^*\theta$, where $\theta=\langle
  E\psi;.\rangle/2$ is a potential 1-form of $\omega$,
  i.e. $\omega=\di \theta$.  By writing $\psi=\coS(p,q,\phi)$, one
has
\begin{align}
\label{i.25}
\dip \psi
k=e^{q^jJA_j}\left(\dip{\eta_p}k+\dip{\Pip}k\phi\right)
\ ,\quad \frac{\partial  \psi}{\partial q^k}
=JA_ke^{q^jJA_j}\left(\eta_p+\Pip\phi\right)
\ ,\quad
(\di_\phi \coS)\Phi=\Pip\Phi
\end{align}
so, taking $\theta=\frac{1}{2}\langle E\psi;.\rangle$, one has
$$
(\coS^*\theta)(P,Q,\Phi)=\frac{1}{2}\left\langle E\psi;\dip \psi j
\right\rangle P\relax_j+
\frac{1}{2}\left\langle E\psi; \frac{\partial  \psi}{\partial q^k}
\right\rangle Q^k+
\frac{1}{2}\left\langle E\psi;(\di_\phi
\coS)\Phi\right\rangle \ .
$$
First we compute the term $\left\langle E\psi;\dip \psi j
\right\rangle$, which coincides with
\begin{align}
\label{i.25b.1}
2\theta\left(\dip \psi j\right)
&=\left\langle
E(\eta_p+\Pip\phi); \dip{\eta_p}j+\dip{\Pip}j\phi\right\rangle
\\
\label{i.25b}
&=\left\langle E\eta_p;\dip {\eta_p}j\right\rangle+
\left\langle E\eta_p;\dip{\Pip}j\phi\right\rangle +
\left\langle E\Pip\phi;{\dip {\eta_p}j}\right\rangle +
  \left\langle E\Pip\phi;\dip{\Pip}j\phi\right\rangle\ .
\end{align}
Now, the third term of \eqref{i.25b} vanishes due to the definition of
$\Pip$. Concerning the first term, there exists a function $f_0$ such
that $\dip{f_0}j= \left\langle E\eta_p;\dip {\eta_p}j\right\rangle$. Indeed, one
has
$$ \dip\null j\left\langle E\eta_p;\dip {\eta_p}i\right\rangle = \dip\null i \left\langle
E\eta_p;\dip {\eta_p}j \right\rangle\ .
$$
Finally, defining $\displaystyle{f_1(p,\phi)=\left\langle E\eta_p;\Pip \phi\right\rangle}$, the second term of \eqref{i.25b} turns out to be
given by $\dip{f_1}j$, so we have
$$
2\theta\left(\frac{\partial \psi}{\partial p_j}\right)=\left\langle
E\Pip\phi;\dip{\Pip}j\phi\right\rangle+\frac{\partial(f_0+f_1)}{\partial p_j}\ .
$$
We compute now $\langle E\psi; (\di_\phi \coS)\Phi\rangle$. We have
$$ 2\langle E\psi; (\di_\phi \coS)\Phi\rangle=\langle
E(\eta_p+\Pip\phi);\Pip\Phi\rangle= \langle
E\Pip\phi;\Pip\Phi\rangle+(\di_{\phi}f_1)\Phi\ .
$$
Adding the easy computation of $\theta(\partial \psi/\partial q^k)$
one gets $\coS^*\theta=\Theta+\di(f_0+f_1)$, and therefore
$\Omega\equiv \coS^* \omega = \di\Theta$.

The explicit computation of $\di\Theta$ gives equation
\eqref{symplsoliton}.
\end{proof}

In order to transform the symplectic form $\Omega$ into the symplectic
form $\Omega_0$ defined in \eqref{sym.ref} we look for a map $\cD$ such
that $\cD^* \Omega = \Omega_0$ in a neighborhood of the soliton
manifold $\Tr_0$. We look for $\cD$ as the time 1 flow $\left. \cF_t
\right|_{t=1}$ of a vector field $Y_t$, asking that $\frac{d}{dt}
\cF_t^* \Omega_t =0$, where $\Omega_t:=\Omega_0 + t (\Omega -
\Omega_0)$. It is well known that the vector field $Y_t$ has to solve
the equation
\begin{equation}
\label{sistema1}
 \Omega_t(Y_t, \cdot) = \Theta_0 - \Theta \ ,
\end{equation}
where $\Theta_0:=-p_j\di q^j+\langle E\phi_j;.\rangle/2$ is such that
$\di\Theta_0=\Omega_0$.  In the next lemma we will study the more
general equation $\Omega_t(Y_t ; .)=\cW$, with $\cW$ an arbitrary
1-form.  Denote $\Omega_t=\langle O_t.;.\rangle$, $\cW=\langle
W,. \rangle$; denote also
\begin{equation*}
Y_t=\left[\begin{matrix}Y_p\\ Y_q\\ Y_\phi
\end{matrix}\right]\ ,\quad W=\left[\begin{matrix} W_{p} \\ W_{q}\\ W_{\phi}
\end{matrix}\right]\ ,
\end{equation*}
where $Y_p \equiv (Y_p^k)_{1 \leq k \leq 4}$ is a vector in $\R^4$,
and similarly $Y_q, W_p, W_q$, while $Y_\phi$ and $W_\phi$ are vectors
in the scale $\cV$.  In this notation, system \eqref{sistema1} takes
the form $O_t Y_t = W$, which in components is given by
\begin{equation}
 \label{systemdarboux}
 \left\lbrace
\begin{aligned}
 & Y_q^k + t\left( \la E \Pipj \phi, \Pipk \phi \ra Y^j_p +
\frac{1}{2}\frac{\partial}{\partial p_k} \la A_j \Pip \phi, \Pip \phi \ra Y^j_q - \la E\Pipk \phi, \Pip Y_{\phi}\ra  \right)= W_{p}^k \\
&- Y_p^k + t \left(- \frac{1}{2} \frac{\partial}{\partial p_j} \la A_k \Pip \phi, \Pip \phi \ra Y^j_p -\la A_k \Pip \phi, \Pip Y_{\phi} \ra \right) =  W_{q}^k \\
&EY_\phi + t \left( (\Pipz^* \Pip^* E \Pip - E)Y_\phi + Y^j_p \Pipz^* \Pip^* E  \Pipj \phi + Y^j_q \Pipz^* \Pip^* A_j \Pip \phi \right)= W_\phi
\end{aligned}
\right.
\end{equation}
The properties of the solution of this system are given in the next lemma.
\begin{lemma}
\label{lem:sol.darboux}
Fix $\cW$ and consider  system \eqref{systemdarboux}.
Then its solution $Y$ is given by
\begin{align}
\label{eq:sol.darboux.1}
Y_p^k &= - (1+ M_1^k(p, \phi))(W_{q}^k + t\la A_k \phi, JW_{\phi} \ra
)+M_3^k(p,\phi) W_p^k  + P_1(t, \phi, W_{\phi})
\\
\label{eq:sol.darboux.2}
Y_q^k &=  (1+ M_2^k(p, \phi))W_{p}^k + M_4^k(p,\phi)(W_{q}^k + t\la
A_k \phi, JW_{\phi} \ra )+ Q_1(t, \phi, W_{\phi})
\\
\label{eq:sol.darboux.3}
Y_\phi &= J W_{\phi} - t Y_q^j J \Pipz^*  A_j \phi + S(t,p, W_{\phi}) +
t Y_q^j \Upsilon_{j}(t,p,\phi) + tY_p^j \Xi_{j}(t,p, \phi)
\end{align}

where the maps $M_1, M_2, M_3, M_4, P_1, Q_1, S$ are smooth in a
neighborhood of $\Tr_0$ in $ \R^4 \times \cV^{-s, -k}$ and fulfill
$\forall s,k \geq 0$
\begin{align*}
& |M_j(p, \phi)| \sleq \norm{\phi}_{\cV^{-s, -k}}^2 \qquad 1 \leq j
  \leq 4 \ ,
\\
& |Q_1(p, \phi, \Phi)| \ , |P_1(p, \phi, \Phi)| \sleq
   |p| \norm{\phi}_{\cV^{-s,-k}} \norm{\Phi}_{\cV^{-s,-k}}
\\
&
  \norm{S(t, p,\phi)}_{\cV^{s,k}} \sleq  |p|
  \norm{\phi}_{\cV^{-s,-k}} \ ,
\end{align*}
and the maps $\Upsilon_{j}$ and $\Xi_{j} $ are of class $C_R(\cK,\cV)$
and fulfill
\begin{align*}
\norm{\Upsilon_j(t,p, \phi)}_{\cV^{s_2, k_2}} \ , \norm{\Xi_{j}(t,p, \phi)}_{\cV^{s_2, k_2}} \sleq  |p| \norm{\phi}_{\cV^{-s_1, -k_1}} \ .
\end{align*}
\end{lemma}

In order to prove Lemma \ref{lem:sol.darboux}, as a first step we will
solve the infinite dimensional equation for $Y_\phi$ (given by the
last component of \eqref{systemdarboux}) as a function of $Y_p$ and
$Y_q$. As a second step, we will substitute such a solution into the
equations for $Y_p$ and $Y_q$, obtaining a finite dimensional system
for $Y_p$ and $Y_q$. The following lemma will be useful:
\begin{lemma}
\label{lem:Dtsmooth}
 The operator $D_t:= E + t(\Pipz^* \Pip^*E \Pip - E)$ is
 skew-symmetric. Furthermore, provided $|p |$ is small enough,
 $D_t$ is invertible and its inverse is given by $D_t^{-1}= J + S_t$
 with
\begin{equation}
\label{A.3.1}
\norma{S_t\phi}_{\cV^{s_2,k_2}}\sleq\norma{\phi}_{\cV^{-s_1,-k_1}}\ ,\quad
\forall s_1,k_1,s_2,k_2\ .
\end{equation}
\end{lemma}

\begin{proof}
 Since $D_t$ acts on $\cV^{s,k}$, we can write $\Pipz^* \Pip^*E \Pip$ as $\Pipz^* \Pip^*E \Pip\Pipz$, from which
skew-symmetry is immediate. We have now that $D_t = E + t \widetilde{D}$ with $\widetilde{D}$ smoothing, since
\begin{equation}
\label{Dtsmooth.coeff}
\Pipz^* \Pip^*E \Pip - E\Pipz = \Pipz^* E (\Pip-\Pipz)
\end{equation}
which is smoothing and fulfills an inequality like \eqref{A.3.1}.
Then $D_t= E(\id + tJ\widetilde{D})$ and inverting by Neumann formula
one gets $D_t^{-1}= J - \sum_{k\geq 1}(-1)^k (tJ\widetilde{D})^k J $
and the thesis.
\end{proof}

\noindent {\em Proof of Lemma \ref{lem:sol.darboux}.}
Consider the last equation of the system \eqref{systemdarboux}.
Introduce
$$ Z_j := \Pipz^* E \Pip \Pipj \phi \ , \qquad \widetilde Z_j :=
\Pipz^* \Pip^* A_j \Pip \phi \ .$$
By Remark \ref{smoothing.proj}, $Z_j$ is in the class $\SH^{1}_{1}$,
while $\widetilde Z_j = \Pipz^* A_j \phi + S^{1}_1 $. Therefore by
Lemma \ref{lem:Dtsmooth} one has
\begin{equation}
\label{solYphi}
\begin{aligned}
 Y_\phi = & D_t^{-1} W_\phi - t Y^j_p D_t^{-1} Z_j - t Y^j_q D_t^{-1}
 \widetilde Z_j = J W_{\phi} + S(t,p, W_{\phi}) - t Y^j_p D_t^{-1} Z_j
 - t Y^j_q D_t^{-1} \widetilde Z_j \ .
 \end{aligned}
\end{equation}
We substitute such a formula in the equations for $Y_p $ and $Y_q$,
obtaining a finite dimensional system for $Y_p$ and $Y_q$. To solve
this system, we need to analyze the regularity of its coefficients. We
begin with the coefficients of the equations in the first line of
system \eqref{systemdarboux}.  Using the definition \eqref{ortvect} of
$\Pip$ one obtains that
\begin{align*}
& \la E \Pipj \phi; \Pipk \phi \ra \in \cR^{0}_2 \ , \qquad
  \frac{\partial}{\partial p_k} \la A_j \Pip \phi; \Pip \phi \ra \in
  \cR^{0}_2 \ .
\end{align*}
Consider now the terms of the form  $\la E \Pipk \phi; \Pip Y_\phi \ra $.
We replace $Y_\phi$ with the  expression \eqref{solYphi}, obtaining that
$$
\la E \Pipk \phi; \Pip Y_\phi \ra = \la E \Pip \Pipk \phi; Y_\phi \ra = a_{jk}(t,p, \phi) Y^j_p  + b_{jk}(t,p, \phi) Y^j_q + c_k(t,p , \phi, W_{\phi})
$$
where the maps $a_{jk}, b_{jk} , c_k$ satisfy for every $s,k \geq 0$
\begin{equation}
\label{darboux.coeff.est}
\begin{aligned}
& \mmod{a_{jk}(t,p, \phi)} \ ,  \mmod{b_{jk}(t,p, \phi)} \sleq  |p| \norm{\phi}_{\cV^{-s, -k}}^2 \ , \\
& \mmod{c_k(t,p, \phi, \Phi)} \sleq  |p| \norm{\phi}_{\cV^{-s, -k}}     \norm{\Phi}_{\cV^{-s, -k}} \ .
\end{aligned}
\end{equation}
We consider now the term $\la A_k \Pip \phi; \Pip Y_\phi \ra$ in the second line of system
\eqref{systemdarboux}.  Inserting the expression \eqref{solYphi}, we have
\begin{align*}
\la A_k \Pip \phi; \Pip Y_\phi \ra  =
& \la  A_k \Pip \phi ; \Pip \Big(J W_{\phi} + S(t,p, W_{\phi}) - t Y^j_p  D_t^{-1} Z_j - t Y^j_q D_t^{-1} \widetilde Z_j \Big) \ra \\
 = & \la A_k \Pip \phi ; \Pip J W_\phi \ra - t \la A_k \Pip \phi ; \Pip \Pipz J A_j \phi \ra Y^j_q  \\  & \qquad + \widetilde c_k(t, p, \phi, W_{\phi}) + \widetilde a_{jk}(t, p, \phi) Y^j_p  + \widetilde b_{jk}(t, p, \phi) Y^j_q  \\
& = \la A_k  \phi ;  J W_\phi \ra + \widetilde c_k(t, p, \phi, W_{\phi}) + \widetilde a_{jk}(t, p, \phi) Y^j_p  + \widetilde b_{jk}(t, p, \phi) Y^j_q
\end{align*}
where to pass from the second to the  third equality  we used that
$$
\la A_k \Pip \phi ; \Pip \Pipz JA_j \phi \ra = \la A_k \phi; J A_j \phi \ra + R^1_2 = R^1_2
$$
as $\la A_k \phi; J A_j \phi \ra =0$.
Once again  $\widetilde a_{jk}, \widetilde b_{jk} , \widetilde c_k $ satisfy  estimates analogous  to \eqref{darboux.coeff.est}.
Altogether we obtain that  system \eqref{systemdarboux} is reduced to the following finite dimensional system:
\begin{equation}
 \label{systemdarboux2}
 \left[ \begin{pmatrix}
                0_{4} & \id_4 \\
                -\id_4 & 0_{4}
               \end{pmatrix}
               +
           K(t,p,\phi)
        \right]
        \begin{pmatrix}
         Y_p \\
         Y_q
\end{pmatrix} \   =              \
\begin{pmatrix}
W_{p} + c(t, p, \phi, W_{\phi}) \\
W_{q} + t\la A_k \phi, J W_{\phi} \ra + \tilde c(t,p,  \phi, W_{\phi} )
\end{pmatrix}
\end{equation}
where $K$ is a matrix whose elements are in the class $\SR^{0}_2$,
while $0_4$ and $\id_4$ denote the zero respectively identity $4
\times 4$ matrices.  The operator on the l.h.s. of \eqref{systemdarboux2} is invertible by
Neumann series, provided $\phi$ belongs to a sufficiently small neighborhood of the origin, and its inverse is given by $\begin{pmatrix} 0_{4} &
  -\id_4 \\ \id_4 & 0_{4}
               \end{pmatrix}   + S_t$, where $S_t$ is a  matrix whose  coefficients are in $\SR^{0}_2$.
Thus we can solve system \eqref{systemdarboux2} and we obtain that
$Y_p$ and $Y_q$ have the claimed structure. Finally we insert the so
find expressions for $Y_p$ and $Y_q$ into equation \eqref{solYphi} and
deduce that also $Y_\phi$ has the claimed structure.  \qed

\vspace{1em}
Now we apply Lemma \ref{lem:sol.darboux} to the case where  $\cW=\Theta_0 - \Theta$:
$$ \cW = -\frac{1}{2} \la E \Pip \phi, \Pipk \phi \ra dp_k +
\frac{1}{2}\la A_k \Pip \phi, \Pip \phi \ra dq^k + \frac{1}{2} \la ( E
- \Pipz^* \Pip^* E \Pip) \phi , \cdot \ra \ .
$$
\begin{remark}
\label{no.q}
In such a case the coefficients $W_p$, $W_q$, $W_\phi$ do not depend
on the variables $q^j$'s, then the vector field $Y$ of Lemma
\ref{lem:sol.darboux} is independent of the $q^j$'s, as one verifies
inspecting the terms in the r.h.s. of \eqref{eq:sol.darboux.1}-\eqref{eq:sol.darboux.3}.
\end{remark}

 \begin{proposition}
\label{yt.lemma}
Let $\cW = \Theta_0 - \Theta$. Then   equation \eqref{sistema1} has a
unique solution $Y_t = [Y_p, Y_q, Y_\phi]$ of the form
\begin{equation}
\begin{aligned}
 Y_{pk} &= -\frac{1}{2}\la A_k \phi, \phi \ra + R^{1}_2(t, p, \phi)  \\
Y_q^k &=   R^{1}_2(t, p, \phi) \ , \\
Y_\phi &= Y_q^k    J \Pipz A_k \phi +  X^{1}_1(t, p, \phi) \ .
\end{aligned}
\end{equation}
Furthermore, $Y_t$ does not depend on the $q^j$'s.
\end{proposition}
\begin{proof}
The 1-form $\cW = \la W; \cdot \ra$ has components given by
\begin{equation}
\label{W.components}
\begin{aligned}
W_p^k = -\frac{1}{2} \la E \Pip \phi; \Pipk \phi \ra  \ , \\
W_{qk} =  \frac{1}{2}\la A_k \Pip \phi; \Pip \phi \ra \ , \\
W_\phi = \frac{1}{2} ( E - \Pipz^* \Pip^* E \Pip) \phi  \ .
\end{aligned}
\end{equation}
One verifies easily that $W_p \in \cR^{1}_2$, $W_\phi \in \cS^1_1 $
and $W_q^k = \frac{1}{2} \la A_k \phi ; \phi \ra + R^1_2$.  Inserting
these expressions into \eqref{eq:sol.darboux.1} gives the following:
\begin{align*}
Y_p^k = & - (1+M_1^k(p, \phi))\left( \frac{1}{2}\la A_k \phi; \phi \ra + R^1_2 \right) + M_3^k (p, \phi) + P^1(t, \phi, S^1_1) \\
= & - (1+M_1^k(p, \phi)) \frac{1}{2}\la A_k \phi; \phi \ra + R^1_2 \\
= & -  \frac{1}{2}\la A_k \phi; \phi \ra - \cP_k(\phi)M_1^k(p,\phi) + R^1_2 \\
= & -  \frac{1}{2}\la A_k \phi; \phi \ra+ R^1_2
\end{align*}
where in the last equality we used that, by the very definition of the class $\cR^i_j$, one
has $\cP_k(\phi)M_1^k(p,\phi) = N_k M_1^k(p, \phi) \in R^1_2$.

Similar computations for the components $Y_q$ and $Y_\phi$ imply  the
claim. Finally by Remark \ref{no.q}, the vector field $Y_t$ does not depend on
the $q^j$'s.
\end{proof}
We are finally able to prove Theorem \ref{darboux}.
\vspace{1em}\\
\noindent{\em Proof of Theorem \ref{darboux}.}  Consider the the
vector field $Y_t$ of Proposition \ref{yt.lemma}. By Lemma \ref{l.6}
it generates a flow $\cF_t$ which is an almost smoothing perturbation
of the identity and such that $\cF_t \vert_{t=1}$ has the structure
\eqref{t.1.1} and provides the wanted change of coordinates. Since
 $Y_t$ does not depend on the $q^j$'s, it follows that the
nonlinear maps in the r.h.s. of \eqref{t.1.1} are independent of the
$q^j$'s as well.

By Remark \ref{timedep}, the inverse transformation has the same
structure.
 \qed

\section{The Hamiltonian after a change of coordinates}\label{coordiHam}

Before computing the Hamiltonian in Darboux coordinates it is worth to
scale the coordinates by introducing the new variables
$(\tilde p,\tilde q,\tilde \phi)$ defined by
\begin{equation}
\label{second}
p=\mu^2 \tilde p\ ,\quad q=\tilde q\ ,\quad \phi=\mu\tilde\phi\ ,
\end{equation}
where
\begin{equation}
\label{mu}
\mu:=\epsilon^{1/4}\ .
\end{equation}
Correspondingly we scale also the $N_k$'s as
$$
N = \mu^2 \tilde N \ .
$$
\begin{remark}
\label{hasca}
The variables \eqref{second} are not canonical, however, under the
change of coordinates \eqref{second}, the Hamilton
equations of a Hamiltonian function $H$ are transformed into the
Hamilton equation of $\tilde H:=H/\mu^2$.
\end{remark}
Due to this scaling, it is convenient to {\bf substitute the classes
  $\cS^{i}_j$ and $\cR^{i}_j$ with new classes $\cS^{i}_j$ and
  $\cR^{i}_j$ in which the order of zeroes in the variable $p$ and $N$
  is substituted by the order of zeroes in $\mu$. }  Thus
$S^k_l(\mu,\tilde p,\tilde q,\tilde \phi)$ will be said to be of class $\cS^k_l$ if there
exists a function $\wtS$ of class $C_R(\R\times\wcK,\cV)$
s.t. $S^k_l( \mu,\tilde p,\tilde q,\tilde \phi)=\wtS^k_l( \mu,\cP(\tilde \phi),\tilde p,\tilde q,\tilde \phi)$ and for
any $s_1,k_1,s_2,k_2$ one has
\begin{equation}
\label{newclass}
\norma{\tilde S^k_l(\mu,\tilde N,\tilde p,\tilde q,\tilde \phi)
}_{\cV^{s_2, k_2}}\sleq \mu^k\norma{\tilde \phi}_{\cV^{-s_1,-k_1}}^l\ .
\end{equation}
\begin{remark}
In order to pass from the "old" classes $\cS^i_j$ to the "new" classes $\cS^i_j$, the following remark is useful:
$$
S \mbox{ in the "old" class } \cS^i_j \Leftrightarrow S \mbox{ in the "new" class } \cS^{2i+j}_j \ .
$$
\end{remark}
When dealing with the scaled variables we will consider again almost
smoothing perturbations of the identity, which are still functions of
the form \eqref{alm}, but with smoothing functions belonging to the
new classes.

{\bf From now on we will only deal with the scaled variables, so we
  will omit the ``tilde'' from the variables}. Furthermore it is
useful also to still denote by $\coS$ the map \eqref{p.q.phi.coord} in
the scaled variables. More precisely, we redefine the map $\coS$
according to
\begin{equation}
\label{res1}
\coS(p,q,\phi)=e^{q^jJA_j}(\eta_{\mu^2 p}+\mu\Pi_{\mu^2p}\phi)\ .
\end{equation}
Similarly we will still denote by $\cD$ the map \eqref{t.1.1} in the
scaled variables \eqref{second}. 

It is worth to remark that, since in the scaled variables the size of the neighborhoods of
$\Tr_0$ one is dealing with is controlled by $\mu$, the open sets (whose existence is ensured by the definitions
of the various classes of objects) can be fixed a priory and the
smallness requirement become just smallness requirements on $\mu$.

Given an open domain $\cU\subset
\cK^{s,k}$ for some $s,k$ and a positive $\rho$, we will denote by
\begin{equation}
\label{ext}
\cU_{\rho}:=\bigcup_{(p,q,\phi)\in\cU}B_\rho^{s,k}(p,q,\phi)\ ,
\end{equation}
where $B_\rho^{s,k}(p,q,\phi)  $ is the open ball in $\cK^{s,k}$ of
radius $\rho$ and center $(p,q,\phi)$.

\begin{remark}
\label{darbdef}
Given an arbitrary neighbourhood $\cU\subset \cK^{1,0}$ of $\Tr_0$,
there exits $\mu_*$ s.t., provided $0\leq\mu<\mu_*$, one has that
$\cD:\cU\to \cK^{1,0}$ is well defined.
\end{remark}

In the scaled variables we will use the following version of Lemma
\ref{l.6}, which will be proved in Appendix \ref{flow.chi}.

\begin{lemma}
\label{l.6.1}
Let $s,P,Q\in \resto^a_j$, $X\in\Sc^a_i $ $j\geq i  \geq 0$, $a\geq 3$
be smoothing functions, and consider the system
\begin{equation}
\label{l.5.12.1}
\dot p=P(\mu,N,p,q,\phi)\ ,\quad \dot q =Q(\mu,N,p,q,\phi)\ ,\quad
\dot \phi= s^l(\mu,N,p,q,\phi)\Pi_0 JA_l\phi+X(\mu,N,p,q,\phi)\ .
\end{equation}
Fix a neighborhood $\U\subset \cK^{1,0}$ of $\Tr_0$ and a positive
$\rho$ in such a way that $s,P,Q,X$ fulfill the estimates
\eqref{newclass} with a constant uniform over the domain $\cU_\rho$
(and depending on $s_1,k_1,s_2,k_2$). Then there exists a positive
$\mu_*$ s.t., provided $0\leq \mu<\mu_*$, the flow $\cA_t$ exists in
$\U$ for $|t|\leq 1$, and is an almost smoothing perturbation of the
identity of class $\ap^3_{j,i}$, namely
\begin{equation}
\label{flt.1}
\cA_t=\Big( p+\bar P(\mu,N,p,q,\phi,t),q+ \bar
Q(\mu,N,p,q,\phi,t), \Pipz\e^{\alpha^l(\mu,N,p,q,\phi,t)JA_l}(\phi+S(\mu,N,p,q,\phi,t))\Big)
\end{equation}
with $\bar P,\bar Q,\alpha^l\in\resto^a_j$ and
$S\in\Sc^a_{i}$.  One also has
$$N(t)=N+R^a_{i+1}+R^{2a}_{2i}\ ,$$ and $\cA_t(\cU)\subset \cU_\rho$. 
Finally  $\bar P,\bar Q,\alpha, S$ fulfill inequalities of the form
\eqref{newclass} with constants uniform on $\cU$.
\end{lemma}

The main result of the present section is the following lemma

\begin{lemma}\label{scaled}
Let $H$ be the Hamiltonian \eqref{H}. Let $\cA\in \ap^3_{0,0}$ be an
almost smoothing perturbation of the identity of the form
\eqref{flt.1}.  Then
$(H \circ\coS\circ\cD\circ \cA^3)/\mu^2$ has the following form:
\begin{equation}
\label{Hsca}
\frac{H \circ\coS\circ\cD\circ \cA^3}{\mu^2}=H_{L}+\mu^2\fh_m+ H_R
\end{equation}
where
\begin{align}
\label{HL}
H_L(\phi)&=\frac{1}{2}\left\langle EL_0\phi;\phi\right\rangle
\\
\label{mech}
\fh_m(p,q)&=\frac{|\bp|^2}{2m}+V^{eff}_m(\bq)\ ,
\qquad \bp=(p_1,p_2,p_3) \ , \bq = (q_1, q_2, q_3) \ , \\
\label{rest}
H_R(\mu,N,p,q,\phi)&= \mu^2 D(N,p)+\frac{1}{2}\left\langle
W^2_0\phi;\phi\right\rangle+\frac{\mu^4}{2}\left\langle V_q
\phi;\phi \right\rangle+ \mu H^3_P(\eta_0+S_{0, hom}^2;\phi+S^2_1)+
\\
&+R^4_{0,hom} +R^3_1+R^2_2\ .
\end{align}
Here
\begin{align}
\label{VH3}
V_q(x)&:=V(x+q)\ ,
\\
\label{VH4}
H^3_P(\eta;\Phi)&:=H_P(\eta+\Phi)-H_P(\eta)-\di
H_P(\eta)\Phi-\frac{1}{2}\di^2H_P(\eta)[\Phi,\Phi]\ ,
\end{align}
$D$ is a smooth function vanishing for $N=0$, $W_0^2$ is a linear
operator of the form
$$
W_0^2[{\rm Re}\phi+\im{\rm Im}\phi]=S^2_0{\rm Re}\phi+S^2_0{\rm Im}\phi\ ,
$$ with different functions $S^2_0$. The quantities $S_{0, hom}^2$,
and $R_{0,hom}^4$ are  functions of class $\cS_{0}^2$ and
$\cR^4_{0,hom} $ respectively,
which are homogeneous of degree $0$ in $\phi$ (they do not depend on
$\phi$).
\end{lemma}

\begin{remark}
\label{domas}
Fix an open neighbourhood $\cU\subset \cK^{1,0}$ of $\Tr_0$, and
assume that the various functions defining $\cA^3$ fulfill the
estimates of the form \eqref{newclass} with constants uniform over
$\cU_\rho$ with some positive $\rho$; then all constants in the
estimates of the form \eqref{newclass} fulfilled by the smoothing
function involved in \eqref{Hsca}-\eqref{VH4} are uniform on the
domain $\cU$.
\end{remark}

\begin{remark}
\label{id}
Since the identity map is of class $\ap^3_{0,0}$, the Hamiltonian in Darboux
coordinates has the form \eqref{Hsca} above.
\end{remark}

\begin{remark}
Let $R^i_j \in \cR^i_j$ and $\cA^k \in \ap^k_{0,0}$. Then $R^i_j \circ \cA^k$ is smoothing and furthermore
$$
R^i_j \circ \cA^k = R^i_j + R^{i+k}_0\ .
$$
\end{remark}
\begin{remark}
Every smoothing map $R \in \cR^i_0$ can be decomposed as
$$
R = R^i_{0,hom} + R^i_1 \ ,
$$ where $R^i_{0,hom} := R\vert_{\phi = 0}$ is in the class $\cR^i_0$
and is homogeneous of degree 0 in $\phi$, while $R^i_1 = R -
R^i_{0,hom} \in \cR^i_1$.
\end{remark}

The rest of the section will be devoted to the proof of Lemma \ref{scaled}
which follows closely the proof of Proposition 2 of \cite{bambusi2013}.

\proof First we remark that $\psi =\left(
\coS\circ\cD\circ\cA^3\right)(p,q,\phi)$ is given by
\begin{equation}
\label{strua}
\psi=\ex{(q^j+R^4_2+R^3_0)}\left(\eta_{\ps}+\mu\Pis\Piz\ex{s^j}\left(\phi+S^2_1+S^3_0\right)\right)\ ,
\end{equation}
where
\begin{equation}
\label{ps}
\ps:=\mu^2(p-N+R^3_0+ R^2_2)\ ,\quad s^j\in\cR^4_2+\cR^3_0\ .
\end{equation}

Substituting in $H_{Free}:=H_0+H_P$ and expanding in Taylor series up
to order three with center at $\eta_\ps$ one gets
\begin{align}
\label{f.1}
H_{Free}(\eta_{\ps})+\di
H_{Free}(\eta_{\ps})\mu\Pis\Piz\ex{s^j}\left(\phi+S^2_1+S^3_0\right)
\\
\label{f.2}
+\mu^2 H_0\left[\Pis\Piz\ex{s^j} \left(\phi+S^2_1+S^3_0\right)  \right]
\\
\label{f.3}
+\frac{\mu^2}{2} \di^2H_P(\eta_{\ps})\left[\Pis\Piz\ex{s^j}\left(\phi+S^2_1+S^3_0\right)
  \right]^{\otimes2}
\\
\label{f.4}
+\mu^3
H^3_P\left(\eta_{\ps};\Pis\Piz\ex{s^j}\left(\phi+S^2_1+S^3_0\right)\right)
\ .
\end{align}
We analyze now line by line this formula. We remark that the second
term of \eqref{f.1} is given by
\begin{align*}
    &\la -\Delta \eta_{\ps}, \Pis \Phi  \ra + \di
  H_{P}\left(\eta_{\ps}  \right) \left[  \Pis \Phi\right]    =
  \lambda^j(\ps) \la A_j \eta_{\ps}, \ \Pis \Phi \ra = 0
\end{align*}
where we used equation \eqref{gro} and the skew-orthogonality
of $\eta_{\ps}$ and $\Pis \Phi$.

In order to compute the first term of \eqref{f.1} denote
$f(p):=H_{Free}(\eta_p)$ and expand
\begin{align}
\label{H_0b}
f(\ps) = f(0) + \sum_{j=1}^4\frac{\d f}{\d p_j}(0) \ps_j
+\sum_{1 \leq j,k \leq 4}\frac{1}{2}\frac{\d^2 f}{\d p_j\d p_k}(0)\ps_k \ps_j+R^6_0
        \ ,
    \end{align}
where we just used that $\ps = \cO( \mu^2)$. Now, one has
$$ \frac{\partial f}{\partial p_j}(p) = \di H_{Free}(\etap)
\frac{\partial \etap}{\partial p_j} = \lambda^k(p) \la A_k \etap,
\frac{\partial \etap}{\partial p_j} \ra = \lambda^k(p) \delta_{j,k} =
\lambda^j(p) \ ,
$$
where $\lambda_j(p)$ are defined in \eqref{lambda}.
Thus it follows
$$
    \frac{\d f}{\d p_j}(0) = 0, \   j =1,2,3 \ ,  \qquad \frac{\d f}{\d p_4}(0) = - \cE(m) \ ,
    $$
    and
    \begin{equation}
    \frac{1}{2} \frac{\partial^2 f}{\partial p_j \partial p_k}(0) =
    \begin{cases}
    \frac{1}{2m}\delta_{j,k} , & 1 \leq j,k \leq 3 \\
    0 , & 1 \leq j \leq 3, \ k=4 \\
    - \frac{1}{2}\cE'(m) , & j=k=4
    \end{cases} \ .
    \end{equation}
    Thus the r.h.s. of \eqref{H_0b} can be written as
    \begin{align*}
\left[ \frac{\sum_{j=1}^{3}\ps_j^2}{2m}-\frac{1}{2}\cE'(m)\ps_4^2
  \right] -\cE(m) \ps_4 +R^6_0
\\
=\mu^4\left[ \frac{|\bp|^2}{2m}-\frac{ \bp\cdot \bN}{m}+
  \frac{|\bN|^2}{2m} - \frac{1}{2}\cE'(m)(p_4^2-2p_4N_4+N_4^2) + R^3_0  + R^2_2
  \right]
\\
-\mu^2\cE(m)(p_4-N_4+R^3_0+R^2_2) +R^6_0
\\
=\mu^4   \left[ \frac{|\bp|^2}{2m}+D(N,p)  \right]+
\mu^2\cE(m) N_4
+R^5_0+R^4_2 \ ,
    \end{align*}
where we defined
$$ D(N,p):=-\frac{\bp\cdot \bN}{m}+ \frac{|\bN|^2}{2m}- \frac{1}{2}\cE'(-2p_4N_4+N_4^2)\ ,
$$
and we omitted terms depending only on $p_4$ which is an integral of
motion for the complete Hamiltonian.

In order to analyze the remaining terms, remark first that
\begin{equation}
\label{scam}
\Pis\Piz\ex{s^j}\left(\phi+S^2_1+S^3_0\right)=\ex{s^j}
\left(\phi+S^2_1+S^3_0\right) \ ,
\end{equation}
where of course the smoothing maps in the two sides of the equality above
are different.  Using \eqref{scam}, the term \eqref{f.2} is easily
seen to be given by
$$
\mu^2 H_0(\phi)+R^4_2+R^5_1+R^8_0\ .
$$
Concerning \eqref{f.3}, it coincides with
\begin{align*}
\frac{\mu^2}{2} \di^2H_P(\ex{-s^j}\eta_{\ps})\left[
  \phi+S^2_1+S^3_0\right]^{\otimes 2}
\\
= \frac{\mu^2}{2} \di^2H_P(\ex{-s^j}\eta_{\ps})(\phi)+R^4_2+R^5_1+R^8_0\ .
\end{align*}
Remark that one has
\begin{equation}
\label{pss}
\ex{-s^j}\eta_{\ps}=\eta_{\ps}+S^3_0+S^4_2=\eta_0+S^2_0\ ,
\end{equation}
so that, taking into account the explicit form of $H_P$, \eqref{f.3} takes the form
\begin{equation}
\label{f.3.f}
\eqref{f.3} = \mu^2\di^2H_P(\eta_0)(\phi)+\frac{\mu^2}{2}\langle
S^2_0\phi;\phi\rangle+ R^4_2+R^5_1+R^8_0\ .
\end{equation}

\begin{remark}
\label{rem:HL}
One has
$$
H_0(\phi)+\frac{1}{2}\di^2H_P(\eta_0)[\phi, \phi]+\cE N_4=\frac{1}{2}\left\langle
EL_0\phi;\phi\right\rangle \equiv H_L(\phi)\ .
$$
\end{remark}

We come to \eqref{f.4}.

First we have that  $H^3_P(\eta; e^{s^j J A_j}\Phi) =
H^3_P(e^{-s^j J A_j}\eta; \Phi)$, which, using \eqref{scam} and \eqref{pss} gives
$$
H^3_P(\eta_\ps; e^{s^j J A_j}\Pis \Pipz(\phi + S^2_1 + S^3_0) = H^3_P( \eta_0 + S^2_0; \phi + S^2_1 + S^3_0) \ .
$$ Now write $S^2_0 = S^2_{0,hom} + S^2_1$ where $S^2_{0,hom}:=
S^2_0\vert_{\phi =0}$ is homogeneous of degree 0 in $\phi$. Exploiting
the definition \eqref{VH4} of $H^3_P$ one has
\begin{align}
H^3_P(\eta_0 + S^2_0; \phi + S_1^2+S_0^3)
\\
\label{f.l}
= H_P(\eta_0 + S^2_{0,hom} + \phi + S^2_1)
- H_P(\eta_0 + S^2_{0,hom}) - \di H_P(\eta_0 + S^2_{0,hom})[\phi + S^2_1]
\\
\label{f.l.1}
-
\frac{1}{2}\di^2 H_P(\eta_0 + S^2_{0,hom})[\phi + S^2_1]^{\otimes 2}
\\
\label{f.l.2}
+ H_P(\eta_0 + S^2_{0,hom}) - H_P(\eta_0 + S^2_{0,hom}+S^2_1)
\\
\label{f.l.3}
+ \di H_P(\eta_0 + S^2_{0,hom})[\phi + S^2_1] - \di H_P(\eta_0 + S^2_{0,hom}+S^2_1)[\phi + S^3_0 + S^2_1]
\\
\label{f.l.4}
\frac{1}{2}\di^2 H_P(\eta_0 + S^2_{0,hom})[\phi + S^2_1]^{\otimes 2} -
\frac{1}{2}\di^2 H_P(\eta_0 + S^2_{0,hom}+S^2_1 )[ S^3_0 + \phi + S^2_1]^{\otimes 2}
\ .
\end{align}
 Now we analyze each line separately. The lines \eqref{f.l}
and \eqref{f.l.1} form the definition of
$$H^3_P(\eta_0 + S^2_{0,hom}; \phi +
S^2_1) \ .
$$
The line \eqref{f.l.2} is a smoothing function in $\cR^2_1$.
The line \eqref{f.l.3} equals
$$
\la \grad H_P(\eta_0 + S^2_{0,hom})- \grad H_P(\eta_0 + S^2_{0,hom} + S^2_1) , \phi + S^2_1 \ra - \la \grad H_P(\eta_0 + S^2_{0,hom}+S^2_1),  S^2_1 + S^3_0 \ra = R^{2}_1 + R^{3}_0 \ .
$$
To analyze the line \eqref{f.l.4} we represent $\di^2H_P(\eta)$ by a
linear operator $W(\eta)$:
$$
\di^2H_P(\eta)(\phi,\phi)=\langle W(\eta)\phi;\phi\rangle\ ,
$$
where explicitly
$$
W(\eta) \Phi := -\beta'(|\eta|^2) \Phi - \beta''(|\eta|^2) |\eta|^2
{\rm Re} \Phi \ .
$$ By smoothness we have that \eqref{f.l.4} is given by
\begin{align*}
 \frac{1}{2} \la \left(W(\eta_0 + S^2_{0,hom})- W(\eta_0 + S^2_{0,hom}+S^2_1)\right) (\phi +
  S^2_1), (\phi + S^2_1)\ra
\\
- \la W(\eta_0 + S^2_{0,hom}+S^2_1) S^3_0, \phi + S^2_1 \ra -
  \frac{1}{2} \la W(\eta_0 + S^2_{0,hom}+S^2_1) S^3_0, S^3_0 \ra \\  \quad = \frac{1}{2} \la
  W^{2}_0 (\phi + S^2_1), (\phi + S^2_1)\ra + R^{3}_1 + R^{6}_0
  \ .
\end{align*}
Thus \eqref{f.4} is equal to
\begin{align*}
\mu^3 H^3_P\left(\eta_0 + S^2_{0,hom};\phi+S^2_1\right) +R^6_0+R^5_1
+\mu^3\left[\frac{1}{2}\langle W^2_0(\phi+S^2_1);\phi+S^2_1\rangle
  +R^3_1+R^6_0  \right]
\\
=\mu^3 H^3_P\left(\eta_0+S^2_{0,hom};\phi+S^2_1  \right) +\mu^3
\frac{1}{2}\langle W^2_0\phi;\phi\rangle +R^6_0+R^5_1\ .
\end{align*}

This concludes the computation of $H_{Free}$.

We come to the simpler computation of $H_V$. First remark that
\begin{equation}
\label{vtilde}
\left[\ex{-(q^j+R^4_2+R^3_0)}V\right](x)= \widetilde V_q(x)=
V_q(x)+W_q(x) ( R^4_2+R^3_0)\ ,
\end{equation}
where
$$ V_q(x):=V(x+q)\ ,\quad W_q(x):=\int_0^1 V'(x+q+\tau(R^4_2+R^3_0))\di\tau \ .
$$
Thus we have
\begin{align*}
H_V=\frac{\mu^4}{2} \langle V_q\eta_{\ps};\eta_{\ps}\rangle
+\frac{\mu^4}{2} \langle W_q\eta_{\ps};\eta_{\ps}\rangle  (
R^4_2+R^3_0)
\\
+\mu^4 \left\langle \widetilde V_q\eta_{\ps};
\mu\Pis\Piz\ex{s^j}\left(\phi+S^2_1+S^3_0\right) \right\rangle
\\
+\frac{\mu^6}{2}\left\langle \widetilde V_q
\Pis\Piz\ex{s^j}\left(\phi+S^2_1+S^3_0\right);\Pis\Piz\ex{s^j}\left(\phi+S^2_1+S^3_0\right)\right\rangle\ .
\end{align*}
Using \eqref{pss}, the first term is easily transformed into $\mu^4
V^{eff}_m(q)+R^6_0$. All the other terms are easily analyzed and give
rise to smoothing terms, except one term coming from the last line,
which is easily seen to produce a term of the form
$$
\frac{\mu^6}{2}\langle \widetilde V_q\phi;\phi\rangle=
\frac{\mu^6}{2}\langle V_q\phi;\phi\rangle+\frac{\mu^6}{2}\langle
W_q\phi;\phi\rangle (R^4_2+R^3_0)\ ;
$$ Including the last term in $\langle W^2_0\phi;\phi\rangle/2$ and
collecting all the results one gets the thesis. \qed

\section{Normal Form}
\label{section_normal_form}

From now on we will work with the Hamiltonian
\begin{equation}
\label{HD}
H_D:=\frac{H\circ\coS\circ\cD}{\mu^2}\ ,
\end{equation}
with $\coS$ given by \eqref{res1} and $\cD$ written in the rescaled
coordinates.

We are interested in eliminating recursively the coupling between
$\phi$ and the mechanical variables. Precisely we want to eliminate
the terms linear in $\phi$.

\begin{definition}
A function $Z(\mu,N, p, q, \phi)$ of class $ \ALS(\R\times \wcK,
\R)$,  will be said
to be in normal form at order $\fr $, $\fr \in \N$,  if the following holds:
\begin{align}
\di_\phi \left.\frac{\partial^r Z}{\partial \mu^r}
(\mu,N,p,q, \phi)\right|_{\substack{ \phi = 0 \\ \mu = 0}} = 0 \ ,\quad \forall r\leq \fr \ .
\end{align}
The derivatives with respect to $\phi$ have to be computed at constant
$N$, i.e. as if $N$ were independent of $\phi$.
\end{definition}
The main result of this section is the following theorem:
\begin{theorem}
\label{normal.form}
Fix an arbitrary $\fr \geq 3$ and an open neighbourhood $\U\subset
\cK^{1,0}$ of $\Tr_0$. Then, there exists a positive $\mu_{*\fr }$
s.t., provided $0\leq\mu<\mu_{*\fr }$, there exists a canonical almost
smoothing perturbation of the identity $\cT^{(\fr )}\in\ap^3_{0,0}$,
$\cT^{(\fr )}: \cU \to \cK^{1,0} $ such that
\begin{equation}
\label{HM}
H^{(\fr )}:=H_D\circ \cT^{(\fr )}
\end{equation}
 is
in normal form at order $\fr $. Furthermore, denoting
$$
(p',q',\phi')=\cT^{(\fr)}(p,q,\phi)\ ,
$$
there exists $C_1$ s.t. one has
\begin{equation}
\label{defor}
\sup_{(p,q,\phi)\in\cU}\left\| q-q'\right\|\leq
C_1\mu^3\ ,\quad\sup_{(p,q,\phi)\in\cU}\left\| p-p'\right\|\leq
C_1\mu^3\ ,\quad \sup_{(p,q,\phi)\in\cU}\norma{\phi'}_{\cV^{1,0}}\leq
(1+C_1\mu^3) \norma{\phi}_{\cV^{1,0}}\ .
\end{equation}
Finally the transformation $\cT^{(\fr )}$ is invertible on its range
and its inverse still belongs to $\ap^3_{0,0}$.
\end{theorem}

To prove the theorem we proceed by eliminating the terms linear in
$\phi$ order by order (in $\mu$). To this end we will use the method
of Lie transform that we now recall.

Having fixed $r \geq 3$, consider a function $\chi_r\in \cR^r_1$,
homogeneous of degree 1 in $\phi$, which
therefore admits the representation
\begin{equation}
\label{chi.ham0}
\chi_{r}(\mu,N,p,q, \phi) = \langle E \chi^{(r)}(\mu,N,p,q), \phi \rangle
\ ,\quad \chi^{(r)}\in\cS^r_0\ .
\end{equation}
\begin{remark}
\label{spagiu}
The function $\chi^{(r)}$ takes values in
$\cV^\infty$, namley one has $\chi^{(r)}=\Pi_0\chi^{(r)}$.
\end{remark}
We are interested in the case where $\chi^{(r)}$ is homogeneous of
degree $r$ in $\mu$.

By Lemma \ref{l.6.1} the Hamiltonian vector field of $\chi_r$
generates a flow $\Phi^t_{\chi_r}\in\ap^r_{1,0}$.

\begin{definition}
\label{lie.tr}
The map $\Phi_{\chi_r}:=\left.\Phi^t_{\chi_r}\right|_{t=1}$ is called
the Lie transform generated by $\chi_r$.
\end{definition}

\begin{remark}
\label{rem:Tr.comp}
 Let  $3\leq r_1 <\cdots < r_n$ be a sequence of integers and let
 $\chi_{r_1},...,\chi_{r_n}$ be functions as above, then the map
$\cT:=\Phi_{\chi_{r_1}}\circ...\circ\Phi_{\chi_{r_n}}$ is an almost
 smoothing perturbation of the identity of class $\ap^{r_1}_{0,0}$. In
 particular one has that $H_D\circ \cT$ has the form \eqref{Hsca}.
\end{remark}
\begin{remark}
\label{rem.3}
Let $F\in\ALS(\R\times\wcK,\R)$, $F=F(\mu,N,p,q,\phi)$, be an almost
smooth function, and let $\chi_r$ be as above; then
$F\circ\Phi_{\chi_{r}}$ is also an almost smooth function, thus it can
be expanded in Taylor series in $\phi$ and in $\mu$ at any order.
\end{remark}
\begin{remark}
\label{evol}
Let $\chi_r\in \cR^r_1$ be as above, and let $F\in\ALS(\R\times
\wcK,\R)$, then
\begin{equation}
\label{evol1}
\mu^a F\circ\Phi_{\chi_r}= \mu^aF+\cO(\mu^{a+r})\ ,
\end{equation}
thus, if $F$ is in normal form, then $\mu^a F\circ\Phi_{\chi_r}$ is in
normal form at order $r+a-1$.
\end{remark}

\noindent{\it Proof of Theorem \ref{normal.form}}. We assume the
theorem true for $r-1$ and prove it for $r$. Assume
$H^{(r-1)}$ is in normal form at order $r-1$.  Thus it has
the form \eqref{Hsca} with $R_1^3$ which actually belongs to
$\cR^{r}_1$. In particular its part homogeneous of degree $1$ in
$\phi$ admits the representation
\begin{equation}
\label{r-1}
R^3_{1,hom}=\langle E\Psi(\mu,N,p,q);\phi\rangle\ ,\quad \Psi\in\cS^{r}_0\ .
\end{equation}

Consider now a function $\chi_r$ as above, and let $\Phi_{\chi_{r}}$
be the corresponding Lie transform. In order to determine $\chi_r$ we
impose that the part of $H^{(r-1)}\circ\Phi_{\chi_r}$ linear
in $\phi$ and homogeneous of order $r$ in $\mu$ vanishes. So first we
have to compute it. By Remark \ref{evol}, it is clear that the only
contributions to such a part come from $H_L\circ\Phi_{\chi_r}$ and
$R_{1,hom}^3\circ\Phi_{\chi_r}\equiv R_{1,hom}^r\circ\Phi_{\chi_r}$. One has
\begin{align}
\label{half}
H_L\circ\Phi_{\chi_r}=H_L+\left\langle EL_0\phi;
\chi^{(r)}\right\rangle +\left\langle EL_0 \phi;\frac{\partial
  \chi_r}{\partial N_j}JA_j\phi\right\rangle+\cO(\mu^{2r})
\\
\nonumber
=H_L+\langle \phi; EL_0\chi^{(r)}\rangle+\left\langle EL_0\phi;\frac{\partial
  \chi_r}{\partial N_j}JA_j\phi\right\rangle+\cO(\mu^{2r})\ ;
\end{align}
now, it is not difficult to see that
$$
\left\langle EL_0 \phi; JA_j\phi\right\rangle = \la W\phi; \phi \ra 
$$
where $W$ is a linear
operator of the form
$$
W[{\rm Re}\phi+\im{\rm Im}\phi]=V_1{\rm Re}\phi+V_2{\rm Im}\phi\ ,
$$ with different potentials $V_1, V_2$ of Schwartz class.
Thus the third term at the last line of \eqref{half} is not linear in $\phi$, so that the only term linear in $\phi$ and of order $r$ in
$\mu$ is $\langle \phi; E L_0\chi^{(r)}\rangle$.

Concerning $R_{1,hom}^3\circ\Phi_{\chi_r}$, one has
$$
R_{1,hom}^3\circ\Phi_{\chi_r}=R^3_{1,hom}+
\cO(\mu^{2r})=\langle E \Psi(\mu,N,p,q);
\phi\rangle+
\cO(\mu^{2r}) \ .
$$
Thus defining
$$
\Psi^{(r)}:=\mu^r\frac{1}{r!}\frac{d^r\Psi}{d\mu^r}(0)
$$ one has that the part of $H^{(r)}$ linear in $\phi$ and of order
$r$ in $\mu$ is
$$
\langle E L_0\chi^{(r)};\phi\rangle+\langle   E\Psi^{(r)};\phi\rangle\ ,
$$
so that the wanted $\chi^{(r)}$ has to fulfill
\begin{equation}
\label{defchi}
L_0\chi^{(r)}=-\Psi^{(r)}\ \Longrightarrow
\chi^{(r)}=-L_0^{-1}\Psi^{(r)}\ ,
\end{equation}
which is well defined since $L_0^{-1}:\cV^{s,r}\to \cV^{s,r}$
smoothly.

Finally we have to add the control of the size of the domain of
definition of $\cT^{(\fr )}$. To get it we proceed as follows: fix a
positive $\rho$ and consider the sequence of domains
$$
\cU\subset\cU_{\rho}\subset\cU_{2\rho}\subset...\subset\cU_{(\fr+1) \rho}\ ;
$$ then, by Lemma \ref{l.6.1} there exists a sequence $\mu_i$,
$i=1,...,\fr $ s.t., if $0\leq\mu<\mu_i$, then
$\Phi_{\chi_{i}}(\cU_{(\fr -i)\rho})\subset\cU_{(\fr-i+1)\rho}$ and
therefore $\cT^{(\fr)}:\cU\to\cU_{\rho(\fr+1)}$. Finally there exists
$\mu_{\fr +1}$ s.t., if $0\leq\mu<\mu_{\fr+1} $ then $\cD$ is well
defined in $\cU_{(\fr+1) \rho}$. Taking $\mu_{*\fr }:=\min{\mu_i}$ one
gets the thesis. \qed

\section{Estimates}\label{disphi}

First we prove an estimate on $\phi$ valid over long times
and then we use it to conclude the proof.  We take  initial data
$(p_0,q_0, \phi_0)$, fulfilling
\begin{equation}
\label{ini}
\norma{p_0}\leq K_0\ ,\quad \norma{\phi_0}_{H^1}\leq K_0 \mu
\end{equation}
(in the rescaled variables) and arbitrary $q_0$.

First we recall that a pair $(r,s)$ is called (Schr\"odinger)
admissible if
$$ \frac{2}{r}+ \frac{3}{s} = \frac{3}{2} \ , \quad 2 \leq s \leq 6,
\quad r \geq 2 \ .
$$
\begin{lemma}
\label{estiphi}
Fix $T_0>0$ and $ \fr \geq 3 $, assume that there exists $T>0$
s.t. the solution $(p,q,\phi)$ of the Hamilton equations of
the Hamiltonian $H^{(\fr )}\equiv H_D \circ \cT^{(\fr )}$ (cf. Theorem
\ref{normal.form}) fulfill the following estimates
\begin{align}
\label{ap.1}
\norma{\phi}_{L^r_t[0,T]W^{1,s}_x}\leq \mu M_1
\\
\label{ap.2}
\sup_{0\leq t\leq T}\norma {p(t)}\leq M_2\ ,
\end{align}
for any admissible pairs  $(r,s)$; then, provided $M_1$ and
$M_2$ are large enough (independently of $T$ and $\mu$), there exists $\mu_*$
independent of $T$, s.t., provided $0\leq\mu<\mu_*$ and $T<T_0/\mu^{\fr-3}$, one
has
\begin{equation}
\label{ap.3}
\norma{\phi}_{L^r_t[0,T]W^{1,s}_x}\leq \mu \frac{M_1}{2}\ 
\end{equation}
for any admissible pair $(r,s)$.
\end{lemma}

First we fix the domain $\cU$ of definition of the Hamiltonian
$H^{(\fr)}$ (cf. eq. \eqref{HM}), which is also the domain over which
the constants involved in the estimates of the smoothing functions
present in \eqref{Hsca}-\eqref{VH4} are uniform. So we define
\begin{equation}
\label{defdomin}
\cU:=\left\{(p,q,\phi)\in\cK^{1,0}\ :\ \norma{p}\leq 2M_2\ ,\quad
\norma{\phi}_{\cV^{1,0}}\leq M_1\right\}\ ,
\end{equation}
so that all the constants involved in the estimates of the smoothing
functions in \eqref{Hsca}-\eqref{VH4} will depend on $M_1,M_2$ but not
on $\mu$.

In this section all the non written constants will depend on $M_1,M_2$ but not
on $\mu$. Sometimes we will make the constants quite explicit in order
to make things clearer.

As a first step we write the equation for $\phi$.
\begin{remark} Denote $X_P:= J \nabla H_P$. Then
\begin{equation}
\label{X^2_P.def}
X^2_P(\eta; \phi) := X_P(\eta + \phi) - \left( X_P(\eta) + d X_P(\eta)
\phi \right)
\end{equation}
is the Hamiltonian vector field of $H_P^3(\eta; \phi)$, i.e. $X^2_P(\eta; \phi) := J \nabla H^3_P(\eta; \phi)$. This can be seen by writing explicitly the definition of Hamiltonian vector field.
\end{remark}
\begin{lemma}
\label{eqphi}
Define $G(\mu,N,p,q,\phi):=\phi+S^2_1(\mu,N,p,q,\phi)$,
where $S^2_1$ is the function at second argument of $H_P^3$ in
\eqref{rest}. Then the Hamilton equation of $H^{(\fr )}$ for
$\phi$ has the form
\begin{align}
\label{p.1}
\dot \phi&= L_0\phi+\left( \frac{\partial H^{(\fr )}}{\partial
  N_j}\right)\Pipz JA_j \phi+\mu^4\Pipz V_q\phi
+W^2_0\phi+S^2_{1,hom} \\
\label{p.2}
&+\mu J[\di
  G]^* EX_P^2(\eta_0+S^2_{0,hom};G)+S^2_2+\frac{1}{2} \langle
\left(J\nabla_\phi W^2_0\right) \phi;\phi\rangle
\\
\label{p.3}
&+S^{\fr }_{0,hom}\ ,
\end{align}
where, as in Lemma \ref{scaled}, we denoted by $S^2_{1,hom} $ a
quantity which is of class  $\cS^2_{1} $ and is homogeneous of degree
$1$ in $\phi$ and we denoted by $\langle\left(\nabla_\phi
W^2_0\right)\phi;\phi\rangle$ the function defined by
\begin{equation}
\label{nablaw2}
\left\langle
\langle\left(\nabla_\phi
W^2_0\right)\phi;\phi\rangle ;h  \right\rangle=
\langle\left(\di_\phi
W^2_0h\right)\phi;\phi\rangle\ ,\quad \forall h\in\cV^{\infty}\ .
\end{equation}
\end{lemma}
\begin{proof} The only nontrivial term to be computed is the vector field of
$H^3_P(\eta_0+S^2_{0hom};G)$. To compute it just remark that, at fixed
$\eta$, one has
$$
J\nabla_{\phi}(H^3_P\circ G)(\eta,\phi)=J\di G^*(\nabla_\phi
H^3_P)(\eta,G(\phi))=J\di G^*EJ(\nabla_\phi
H^3_P)(\eta,G(\phi))=J \di G^*E X^2_P(\eta;G(\phi))\ ,
$$ and that, since $S^2_{0,hom}$ is independent of $\phi$, the
gradient of $H^3_P$ with respect to the first argument enters in the
equations only through $\frac{\partial H^{(\fr)}}{\partial N_j}$. Then
the result immediately follows.
\end{proof}

To estimate the solution of \eqref{p.1}--\eqref{p.3} consider first the time
dependent linear operator
$$
L(t): =  L_0 +  w^j(t) \Pipz J A_j  \ ,
$$
where
$ w^j(t) = \frac{\partial H^{(\fr )}}{\partial N_j} (p(t), q(t), \phi(t))$.
\begin{remark}
\label{estiw}
Exploiting the inductive assumptions \eqref{ap.1}--\eqref{ap.2}  and computing the explicit form of $w^j$, one has that
  \begin{equation}
\label{estiw1}
\sup_{t\in[0,T]}|w^j(t)|\sleq \mu^2\ .
\end{equation}
\end{remark}

Denote by $\cU(t,s)$ the evolution operator of $L(t)$. The following
lemma was proved in \cite{beceanu, bambusi2013,perelman}:
\begin{lemma}
\label{pere}
Assume \eqref{estiw1}. There exists $C_0$ independent of $M_1,M_2$,
and 
$\mu_*$ (dependent on $M_1,M_2$)  s.t., provided
$0\leq \mu<\mu_*$, the following Strichartz estimates hold
\begin{align}
\label{stric.est}
&\norm{\cU(t,0)  \phi }_{L_t^r W^{1,s}_x} \leq C_0 \norm{\phi}_{H^1}
\ , \\
\label{stric.est.2}
&\norm{ \int_0^t \cU(t,\tau) F(\tau) d\tau }_{L_t^{r} W^{1,s}_x} \leq C_0
\norm{F}_{L_t^{\wtr'} W^{1,\wts'}_x} \ ,
\end{align}
where $(r,s)$ and $(\wtr, \wts)$ are admissible pairs and
$(\wtr',\wts')$ are the exponents dual to $(\wtr,\wts)$.
\end{lemma}

In order to prove Lemma \ref{estiphi} we will make use of the
following Duhamel formula
\begin{align}
\label{phi.1}
\phi(t) = & \cU(t,0)\phi_0
\\
\label{phi.11}
&+ \int_0^t \cU(t,\tau) [\mu^4\Pipz
  V_{q(\tau)}\phi(\tau)+W^2_0\phi(\tau)+S^2_{1,hom}(\tau)]\di \tau \\
\label{phi.2}
 &+\int_0^t \cU(t,\tau)\left[\mu J[\di
    G]^*EX_P^2(\eta_0+S^2_{0,hom};G)+S^2_2+\frac{1}{2} \langle
  \left( J\nabla_\phi W^2_0\right) \phi;\phi\rangle\right]\di \tau  \\
\label{phi.3}
&+ \int_0^t \cU(t,\tau)S^\fr _{0,hom}(\tau)\di \tau \ .
\end{align}

We estimate term by term the argument of the different integrals.

\begin{lemma}
\label{lem:W.est}
One has
\begin{equation}
\label{estilin}
\norma{ \mu^4\Pipz
  V_{q}\phi+W^2_0\phi+S^2_1}_{W^{1,6/5}_x}\sleq \mu^2
\norma\phi_{W^{1,6}_x} \ .
\end{equation}
\end{lemma}
\proof Consider the first term. By Leibniz rule and H\"older
inequality, one has
\begin{equation}
\label{hol}
\norma{
  V_{q}\phi}_{W^{1,6/5}_x}\sleq \norma{
  V_{q} }_{W^{1,3/2}_x}\norma{\phi}_{L^{6}_x}+ \norma{
  V_{q} }_{L^{3/2}_x}\norma{\phi}_{W^{1,6}_x}\ ,
\end{equation}
which gives the estimate of such a term. The second term is estimated
in the same way, while the third one is a trivial consequence of the
definition of smoothing map. \qed
\begin{remark}
\label{w2est}
By recalling that $W_0^2$ multiplies the real and the imaginary parts
of $\phi$ by a smoothing function, one has
\begin{equation}
\label{w2esti.1}
\norma{ \langle \left(J\nabla_\phi W^2_0\right)
  \phi;\phi\rangle}_{W^{1,6/5}_x}\sleq \mu^2 \norma{\phi}_{L_x^2}
\norma{\phi}_{W^{1,6}_x}\sleq \mu^3\norma{\phi}_{W_x^{1,6}}\ .
\end{equation}
\end{remark}

\begin{lemma}
\label{lem:W.est2}
One has
\begin{equation}
\label{estinon}
\norma{ \mu J[\di
    G]^*EX_P^2(\eta_0+S^2_{0,hom};G)+S^2_2}_{W^{1,6/5}_x}\sleq \mu
\norma{\phi}_{W^{1,6}_x} \norma{\phi}_{H^1_x} \left( 1+
\norma{\phi}_{H^1_x}\right)\sleq\mu^2 \norma{\phi}_{W^{1,6}_x} \ .
\end{equation}
\end{lemma}
\begin{proof} We start by estimating the norm of $X^2_P(\eta;\Phi)$, with
arbitrary $\eta$ of Schwartz class. First remark that
$X^2_P(\eta;\Phi)$ is given by ($\Pipz$ applied
to) the function
\begin{equation}
\label{beta.estimate}
\beta'(\left|\eta+\Phi\right|^2)(\eta+\Phi)-\beta'(|\eta|^2)\eta
-\beta'(|\eta|^2)\Phi -\beta''(|\eta|^2)\left|\eta\right|^2(\Phi+\bar
\Phi)
\end{equation}
whose modulus is easily estimated (using also \eqref{num}), obtaining that
$$
\mmod{\eqref{beta.estimate}} \sleq \left[\frac{|\Phi|^2}{\langle x\rangle^{k}}+|\Phi|^3+|\Phi|^5\right]\ ,
$$
with arbitrary $k$. Thus one has
\begin{equation}
\label{estix2}
\norma{X^2_P(\eta;\Phi)}_{L^{6/5}_x}\leq
C[\left\|\Phi\right\|^2_{L^6_x}+\left\|
  \Phi\right\|_{L^2_x}\left\|\Phi\right\|^2_{L^6_x}+\left\|\Phi\right\|_{L^6_x}^5]
\ .
\end{equation}
Exploiting Sobolev embedding theorem one gets that this is controlled
by
\begin{equation}
\label{estix2.1}
\left\|\Phi\right\|_{L^6_x}\left\|\Phi\right\|_{H^1_x} [1+\left\|\Phi\right\|_{H^1_x}^3] \ ,
\end{equation}
and, exploiting Leibniz formula  one also gets
\begin{equation}
\label{estix2.2}
\norma{X^2_P(\eta;\Phi)}_{W^{1,6/5}_x}\sleq
\left\|\Phi\right\|_{W^{1,6}_x}\left\|\Phi\right\|_{H^1_x} [1+\left\|\Phi\right\|_{H^1_x}^3] \ .
\end{equation}
Adding the simple estimate of $G$, $[\di G]^*$
and $S^2_2$ one gets the thesis.
\end{proof}

\noindent{\it End of the proof of Lemma \ref{estiphi}}. Consider the
integral equation \eqref{phi.1}--\eqref{phi.3}. Using the Strichartz
estimates \eqref{stric.est}--\eqref{stric.est.2}, the estimates
\eqref{estilin}--\eqref{estinon}, and the inductive assumptions
\eqref{ap.1}--\eqref{ap.2} one has, writing explicitly the constants  
\begin{align}
\label{stifin}
\norma{\phi}_{L^r_t[0,T]W_x^{1,s}}\leq &
C_0\left\|\phi_0\right\|_{H^1_x} + C(M_1,M_2)\mu^2\left\|\phi \right\|_{L^2_t[0,T]
  W^{1,6}_x} +
\left\|S^\fr _{0,hom}\right\|_{L^1_{t}[0,T]H^1_x} \\ \leq & C_0 K_0
\mu +C(M_1,M_2)\mu^3 +C(M_1,M_2)T\mu^{\fr}\ ,
\end{align}
which is smaller then $\mu M_1/2$ provided one chooses $M_1$ large
enough, $\fr \geq 3$, $0\leq \mu\leq \mu_*$ with $\mu_*$ small enough and $T
< T_0/\mu^{\fr-3}$.\qed

\vspace{1em}
Now we use the dispersive estimates of Lemma \ref{estiphi} to
prove that the quantity
$$
H_{L}(\phi) = \frac{1}{2} \la E L_0 \phi, \phi \ra
$$ is almost conserved for times of order $T_0/\mu^{\fr-3}$. We need the
following preliminary lemma
\begin{lemma}
\label{lem:HL0.est}
Let $X \in C^0(\cU, W^{1,6/5}_x)$ be a vector field, and  let $w^j
\in C^0([0, T], \R)$, $1 \leq j \leq 4$, be functions depending also
on $\mu$ and fulfilling
\begin{align}
\label{X.est.0}
\sup_{\substack{q \in \R^4, \norm{p} \leq M_1 \\ \norm{\phi}_{H^1}\leq \mu M_2}}\norm{X(p,q,\phi)}_{W^{1, 6/5}_x} \sleq  \mu^2 \norm{\phi}_{W^{1, 6}_x} \ ,
\\
\label{w.est.0}
\sup_{t \in [0, T]}|w^j(t)| \sleq  \mu^2 , \qquad \forall j \ .
\end{align}
Then one has
\begin{align}
\label{EL0.X}
\sup_{\substack{q \in \R^4, \norm{p} \leq M_1 \\ \norm{\phi}_{H^1}\leq \mu M_2}} \mmod{\la E L_0 \phi, X(p,q,\phi) \ra } \sleq \mu^2
\norm{\phi}_{W^{1,6}_x}^2 \ ,
\\
\label{EL0.A}
\sup_{t \in [0,T]}\mmod{w^j(t) \la E L_0 \phi, J A_j \phi\ra } \sleq \mu^2
\norm{\phi}_{W^{1,6}_x}^2 \ .
\end{align}
\end{lemma}
\begin{proof}
First we prove \eqref{EL0.X}. Using the specific form of $E L_0 :=
-\Delta + W(\eta_0) + \cE_0$, we get that
\begin{align*}
\mmod{\la E L_0 \phi, X(p,q,\phi) \ra} \leq & \mmod{\la \Delta
  \phi, X(p,q,\phi) \ra} + \mmod{\la W(\eta_0) \phi, X(p,q,\phi)
  \ra} + \mmod{\la \cE_0 \phi, X(p,q,\phi) \ra} \\ \leq &
\norm{\nabla \phi}_{L^6_x}\norm{\nabla X}_{L^{6/5}_x} +
\left(\norm{W(\eta_0)}_{L^{\infty}_x} + |\cE_0| \right)\norm{\phi}_{L^6_x}
\norm{X}_{L^{6/5}_x} \\ \sleq & \mu^2  \norm{\phi}_{W^{1, 6}_x}^2
\ .
\end{align*}
Thus \eqref{EL0.X} is proved. We prove now \eqref{EL0.A}.  Once again
we use the specific form of $L_0$, and the fact that since $\Delta$,
$A_4$ and $A_j$ are self-adjoint commuting operator and $J$ is
skew-symmetric, one has
$$
\la \Delta \phi, JA_j \phi \ra = 0= \la  \phi,  JA_j \phi \ra  \ .
$$
 Thus it follows that
\begin{align*}
\mmod{\la E L_0 \phi, J A_j \phi \ra} =    \mmod{\la W(\eta_0) \phi, J A_j \phi \ra}
\leq  \norm{W(\eta_0) \phi}_{L^{6/5}_x} \norm{\nabla \phi}_{L^6_x} \leq \norm{W(\eta_0)}_{L^{3/2}_x} \norm{\phi}_{W^{1, 6}_x}^2  \ .
\end{align*}
This estimate together with \eqref{w.est.0} implies  \eqref{EL0.A}.
\end{proof}

In the next lemma we show that $H_{L}(t) := \langle E L_0 \phi(t),
\phi(t) \rangle/2$ stays very close to its initial value for large
times.

\begin{lemma}
\label{lem:HLO.diff}
Under the same assumptions of  Lemma \ref{estiphi}, assume  $T
<T_0/\mu^{\fr-3}$ then one has
\begin{equation}
\sup_{t \in [0,T]} \mmod{H_{L}(t) - H_{L}(0)} \sleq  \mu^4  \ .
\end{equation}
\end{lemma}
\begin{proof}
To begin with, we write
$$
H_{L}(t) = H_{L}(0) + \int_0^t \frac{d}{dt} H_{L}(\tau) \, d\tau = H_{L}(0) + \int_0^t \la E L_0 \phi(\tau), \dot \phi(\tau) \ra \, d\tau \ .
$$ Substituting the equations of motion of $\phi$, we obtain that
$\int_0^t \la E L_0 \phi(\tau), \dot \phi(\tau) \ra d\tau = \sum_{j=1}^5
I_j$, where
\begin{align*}
I_1(t) := \int_0^t \la E L_0 \phi(\tau), L_0 \phi(\tau) \ra  d\tau ,
\quad
I_2(t) := \int_0^t \la E L_0 \phi(\tau), w^j(\tau) J A_j \phi(\tau) \ra  d\tau ,
 \\
I_3(t) :=  \int_0^t \la E L_0 \phi(\tau),
\mu^4\Pipz\widetilde
  V_{q(\tau)}\phi(\tau)+W^2_0\phi(\tau)+S^2_1(\tau)
\ra  d\tau ,
\\
I_4(t) := \int_0^t \la E L_0 \phi(\tau), \mu J[\di
    G]^*EX_P^2(\eta_0+S^2_{0,hom};G)+S^2_2 +\frac{1}{2}\langle
\left(J\nabla_\phi W^2_0\right)\phi;\phi\rangle\ra  d\tau ,
\\
I_5(t) := \int_0^t \la E L_0 \phi(\tau), S^\fr _{0,hom}(\tau)  \ra  d\tau .
\end{align*}
By the skew-symmetry of $E$, $I_1 \equiv 0$.  Consider now $I_2$. By
Remark \ref{estiw} the functions $w^j(t)$, $1 \leq j \leq 4$, satisfy
estimate \eqref{w.est.0}. By Lemma \ref{lem:HL0.est} it follows that,
for every $0 \leq t \leq T$,
$$
\mmod{I_2(t)} \leq \int_0^t \mmod{\la E L_0 \phi(\tau), w^j(\tau) J A_j \phi(\tau) \ra } \,  d\tau \sleq  \mu^2   \norm{\phi}_{L^2_t {[0,T]}W^{1,6}_x}^2  \ .
$$
Consider now $I_3(t)$. By Lemma \ref{lem:W.est} the vector field at
r.h.s. of the scalar product satisfies the estimate \eqref{X.est.0},
therefore one has
$$
\mmod{I_3(t)} \sleq \mu^2  \norm{\phi}_{L^2_t[0,T]W^{1,6}_x}^2 \ .
$$
The term $I_4(t)$ is estimated in a similar  way, using Lemma
\ref{lem:W.est2} and Remark \ref{w2est}.

We estimate now $I_5(t)$. Using that $EL_0 = -\Delta + W(\eta_0) + \cE$,
one has
\begin{align*}
\mmod{\la E L_0 \phi,  S^\fr _{0,hom}  \ra} =\mmod{\la\phi,  E L_0
  S^\fr _{0,hom}  \ra} \sleq  \norma{\phi}_{L^2_x}\mu^\fr  \ .
\end{align*}
Inserting this estimate in the expression for $I_5(t)$, one gets that
\begin{align*}
\sup_{t \in [0,T]}|I_5(t)| \sleq
\left\|\phi\right\|_{L^\infty_t[0,T]H^1_x}T\mu^{\fr } \ .
\end{align*}
Altogether we have  that
$$ \sup_{t \in [0,T]}\mmod{H_{L}(t) - H_{L}(0)} \sleq  \mu^2
\norm{\phi}_{L^2_t {[0, T]}W^{1,6}_x}^2 +
\left\|\phi\right\|_{L^\infty_t[0,T]H_x^1}T\mu^{\fr} \ .
$$
Using  estimate \eqref{ap.3} and taking $T <T_0/\mu^{\fr -3}$ one gets the claim.
\end{proof}

We are finally ready to prove that the mechanical energy of the soliton does not change for long times.
\begin{theorem}
\label{thm:h}
Under the same assumptions of Lemma \ref{estiphi}, there exists
$C(M_1,M_2)$ s.t., for $T<T_0/\mu^{\fr
  -3}$, one has
\begin{equation} \sup_{t \in [0,T]}
\mmod{\fh_m(t) - \fh_m(0)} \leq  C(M_1,M_2)\mu^2 \ .
\end{equation}
\end{theorem}
\begin{proof}
Consider $H^{(\fr )}$; by the conservation of energy, one has that
$H^{(\fr )}(t) \equiv H^{(\fr )}(p(t), q(t), \phi(t))=H^{(\fr )}(0)$. Write
$H^{(\fr )} = \mu^2\fh_m + H_L + H_R$ (as in \eqref{Hsca}), and remark
that, under the inductive assumptions \eqref{ap.1}, \eqref{ap.2},
$\left|H_R(t)\right|\leq C\mu^4$ so that one has for every $0 < t < T< T_0/\mu^{\fr -3}$
\begin{align*}
\mu^2 \mmod{\fh_m (t) - \fh_m (0)} \leq \mmod{H_{L} (t) - H_{L} (0)} +
\left|H_R(t)\right|+\left|H_R(0)\right|\leq C(M_1,M_2)\mu^4 \ .
\end{align*}
\end{proof}

The last step is to show that the inductive assumption \eqref{ap.2} holds. This is provided by
the following lemma.

\begin{lemma}
\label{cp}
Assume that \eqref{ap.1}, \eqref{ap.2} hold. Then, provided $M_2$ is
large enough, one has that, provided
$T<T_0/\mu^{\fr -3}$, one has
\begin{equation}
\label{stifin.p}
\sup_{t\in[0,T]}\norma{p(t)}\leq \frac{M_2}{2}\ .
\end{equation}
\end{lemma}
\begin{proof} First remark that $p_4$ is an integral of motion, then just use
the form of $\fh_m$, namely
$$
\fh_m(p,q)=\frac{|\bp|^2}{2m}+V^{eff}_m(\bq)\ ,
$$ the fact that $V^{eff}_m$ is globally bounded to get 
\begin{equation}
\label{stifp}
\frac{\left\|\bp(t)\right\|^2}{2m}\leq
\fh_m\big|_{t=0}+C(M_1,M_2)\mu^2
+\sup_{\bq\in\R^3}\left|V^{eff}_m(\bq)\right| \leq \frac{K_0^2}{2m}+
C(M_1,M_2)\mu^2 +2\sup_{\bq\in\R^3}\left|V^{eff}_m(\bq)\right|\ ,
\end{equation}
which is smaller than 
$$
\left(\frac{M_2}{2}\right)^2\frac{1}{2m}\ ,
$$ provided $M_2$ is sufficiently large and $\mu$ sufficiently small.
\end{proof}

So (changing $\fr$ to $\fr+3$) we have obtained the following lemma.

\begin{lemma}
\label{lem:stifin}
Fix $K_0$, $T_0$, and $\fr$, then there exists positive $\mu_{\fr}$,
$M_1,M_2$ s.t., provided $0\leq \mu<\mu_{\fr}$, the following holds true:
assume that the initial data fulfill
\begin{equation}
\label{kzero}
\norma{\phi_0}_{\cV^{1,0}}\leq \mu K_0\ ,\quad \norma{p_0}\leq K_0\ ,
\end{equation}
then, along the corresponding solution one has
\begin{equation}
\label{stifina}
\norma{\phi(.)}_{L^r_t[0,T_0/\mu^{\fr}
  ]W^{1,s}_x}\leq \mu M_1\ , \quad\norma{p(.)}_{L^\infty_t[0,T_0/\mu^{\fr}]}\leq
M_2\
\end{equation}
for any admissible pair $(r,s)$.  Furthermore there exists $K_3$
s.t. one has
\begin{equation}
\label{esm}
\sup_{0\leq t\leq T_0/\mu^{\fr}}\left| \fh_m(t)-\fh_m(0) \right|\leq
K_3\mu^2\ .
\end{equation}
\end{lemma}

\vspace{1em}

We conclude this section with the proof of Theorem \ref{main} and
Corollary \ref{1d}.
\vspace{1em}\\
\noindent {\em Proof of Theorem \ref{main}.} Here it is needed to
distinguish between the variables introduced by $\coS$ through
\eqref{res1}, and the variables obtained after application of Darboux
and normal form theorem. We will denote by $(p,q,\phi) $ the variables
introduced by \eqref{res1}, and by $(p',q',\phi')$ the variables s.t. 
$$
(p,q,\phi)=(\cD\circ\cT^{(\fr)})(p',q',\phi')\ .
$$ We define the functions $\alpha(t)=q^4(t)$, $\bp(t)$ and $\bq(t)$
to be the solutions of the equations of motion in the variables
\eqref{res1} (to get the theorem one actually has to scale back
$\bp$). Remark that with these notations all the preceding part of
this section deals with the variables $(p',q',\phi')$.  So, by
Corollary \ref{bar}, if the initial datum fulfills \eqref{epsilon0},
then in the variables $(p,q,\phi)$ the estimates \eqref{kzero}
hold. Then the same holds in the variables $(p',q',\phi')$ (due to the
definition of the class of $\cD$ and eq. \eqref{defor}). So we can
apply Lemma \ref{lem:stifin} getting the result in the variables
$(p',q',\phi')$. To get the final statement we have to show that it also
holds in the variables $(p,q,\phi)$ just defined. This follows from
\begin{align}
\label{hscalino}
&\left|\fh_m(p(t),q(t))-\fh_m(p(0),q(0))\right| \\ & \leq
\left|\fh_m(p(t),q(t))-\fh_m(p'(t),q'(t))\right|+\left|\fh_m(p'(t),q'(t))-\fh_m(p'(0),q'(0))\right|
\\
&\null\quad+\left|\fh_m(p(0),q(0))-\fh_m(p'(0),q'(0))\right|\sleq
\mu^2\ .
\end{align}
\qed

\noindent
{\it Proof of Corollary \ref{1d}}. First we work in the scaled
variables, cf. eq. \eqref{res1}. In this case the corollary is a
trivial consequence of the fact that, under its assumptions both $\bp$
and $\bq$ are actually one dimensional, so, at any moment they lie on
the curve identified by $\fh_m(t)$. In turn, by estimate
\eqref{hscalino}, such a curve is $O(\mu^2)$ close to the level
surface $\fh_m(0)$ (recall that we are assuming that we are not at
critical point of $\fh_m$). This is true in the standard distance of
$\R^2$. Scaling back the variables to the original variables, one gets
the result. \qed

\appendix

\section{Proof of Lemma \ref{invcoo}.} \label{invc}

First we prove a local result close to $0$.

\begin{lemma}
\label{l.con.coo}
There exists a mapping $\vphi(\psi )\equiv (p(\psi ),q(\psi ))$ with the
following properties
\begin{itemize}
\item[1)] $\forall k,s$ there exists an open neighborhood
  $\U^{-k,-s}\subset \cH^{-k,-s}$ of $\eta_{0}$ such that $\vphi\in
  C^\infty (\U_{-k,-s},\R^{2n})$
\item[2)] $e^{-q^j(\psi )JA_j}\psi -\eta_{p(\psi )}\in\Pi_{p(\psi )}
  \V^{-k,-s}$.
\end{itemize}
\end{lemma}
\proof Consider the condition 2). It is equivalent to the couple of
equations
\begin{align}
\label{e.con.coo}
0=f_l(p,q,\psi ):=\langle e^{-q^jJA_j}\psi -\eta_p;A_l\eta_p\rangle\equiv
\langle \psi ;e^{q^jJA_j}A_l\eta_p \rangle-2p_l=0\ ,
\\
0=g^l(p,q,\psi ):= \langle e^{-q^jJA_j}\psi -\eta_p;E\depv{\eta_p}l
\rangle\equiv \langle \psi ; e^{q^jJA_j}E\depv{\eta_p}l  \rangle
-\langle\eta_p ; E\depv{\eta_p}l\rangle
\end{align}
Both the functions $f$ and $g$ are smoothing, so we apply the
implicit function theorem in order to define the functions $q(\psi )$,
$p(\psi )$. First remark that the equations are fulfilled at
$(p,q,\psi )=(0,0,\eta_{0})$, then we compute the derivatives of such
functions with respect to $q^j,p_j$ and show that they are
invertible. We have
\begin{eqnarray*}
\depv{f_j}k\big\vert_{(0,0,\eta_{0})}=\left[\langle \psi ;e^{q^l
    JA_l}A_j\depv{\eta
    _{p}}k\rangle-2\delta_j^k\right]_{(0,0,\eta_{0})}=
-\delta_j^k\ ,
\end{eqnarray*}
where we used
\begin{equation}
\label{e.con.1}
\delta_j^k=\depv{\null}k\frac{1}{2}\langle \eta_{p};A_j\eta_{p}\rangle=
\langle \eta_{p}; A_j \depv{\eta_p}k\rangle\ .
\end{equation}
Then we have
\begin{equation}
\label{e.eq.q}
\frac{\partial f_j}{\partial
  q^k}\big\vert_{(0,0,\eta_{0})}=\langle\eta_{0};
JA_jA_k\eta_{0} \rangle=0
\end{equation}
by the skew-symmetry of $J$.

We come to $g$.
\begin{equation}
\label{e.eq.1}
\depv{g^j}k\big\vert_{(0,0,\eta_{0})}=\langle\eta_{0};
E\frac{\partial^2\eta_{0}
}{\partial p_j\partial p_k}\rangle-\langle\depv {\eta_{0}}k;E\depv{\eta_{0}}j
\rangle   -\langle\eta_{0};
E\frac{\partial^2\eta_{0}
}{\partial p_j\partial p_k}\rangle
\end{equation}
which vanishes. Finally we have
$$
\frac{\partial g^j}{\partial q^k}\big\vert_{(0,0,\eta_{0})}=\langle
A_k\eta_{0};\depv {\eta_{0}}j \rangle=\delta_k^j\ .
$$
Therefore the implicit function theorem applies and gives the result.
\qed

As a corollary one gets that close to $\eta_0$ one can define the map
$$\coS^{-1}(\psi):=(p(\psi),q(\psi),\widetilde{\Pi}_{p(\psi)}^{-1}(
e^{-q^j(\psi )JA_j}\psi -\eta_{p(\psi )} )\ , $$ where the inverse
$\widetilde{\Pi}_{p}$ of $\Pip$ is defined in Remark
\ref{smoothing.proj2}.

Repeating the argument of Lemma \ref{l.con.coo} at an arbitrary point
$e^{q^jJA_j}\eta_{p}$ one gets that the map $\coS$ is a local
homeomorphism (in anyone of the spaces of the scale), close to any
point of $\Tr_0$ and furthermore the size of the ball over which this
holds does not depend on the point of $\Tr_0$. In order to transform it into a
global homeomorphism we have just to identify points with the same
image, which of course are points in which the coordinate $q^4$
differs by $2\pi$. \qed

\section{Proof of Lemma \ref{l.6.1} }
\label{flow.chi}
In order to solve the system \eqref{flt.1}, we introduce some
auxiliary independent variables. In particular we will introduce the
$N$'s as auxiliary variables and we make the change of variables
$
\phi=e^{\alpha^jJA_j}u \ ,
$
and ask the quantities $\alpha^j$ to fulfill the equation $\dot
\alpha^j=s^j$. To get the equation for $N_b$ simply compute
\begin{align*}
\dot N_b= \langle A_b\phi;\dot \phi\rangle= s^j\langle
A_b\phi;JA_j\phi\rangle+
+s^j \langle
A_b\phi;(\Pipz-\id)JA_j\phi\rangle+\langle
A_b\phi;X\phi\rangle= R^a_{i+1}\ .
\end{align*}
Thus the original system turns out to be equivalent to
\begin{align}
\label{flt.5}
\dot p=P(\mu,N,p,q,e^{\alpha^jJA_j}u)\ ,\ \dot
q=Q(\mu,N,p,q,e^{\alpha^jJA_j}u)\ ,\quad \dot \alpha^l=s^l
(\mu,N,p,q,e^{\alpha^jJA_j}u)
\\
\label{flt.6}
\dot N= R^a_{i+1}(\mu,N,p,q,e^{\alpha^jJA_j}u)\ ,\ \dot u
=e^{-\alpha^jJA_j}s^l
(\Pipz-\id)JA_le^{\alpha^jJA_j}u+e^{-\alpha^jJA_j}X\ ,
\end{align}
which is a smooth system in all the spaces of the scale. Thus, by
standard contraction mapping principle, it admits a solution. To
obtain the estimate on the domain, and the fact that the solution
belongs to the wanted classes, just remark that on the domain
$\cU_\rho$ (extended by addition of the auxiliary variables), the
vector field \eqref{flt.5}--\eqref{flt.6} is dominated by a constant
times $\mu^a$ and take into account the degree of homogeneity in
$\phi$ of the various components.\qed

\addcontentsline{toc}{chapter}{Bibliography}
\nocite{*}


\end{document}